\documentclass[amsmath,amssymb]{amsart}

\usepackage{dcolumn}
\usepackage{bm}
\usepackage{graphicx}
\usepackage{amsmath}
\usepackage{amsfonts}
\usepackage{amssymb}
\usepackage{amsthm}
\usepackage{amstext}
\usepackage{amsbsy}
\usepackage{amsopn}
\usepackage{amscd}
\usepackage{amsxtra}

\newtheorem{lemma}{Lemma}
\newtheorem{proposition}{Proposition}
\newtheorem{definition}{Definition}
\newtheorem{corollary}{Corollary}

\newtheorem{remark}{Remark}

\def\C{\mathbb C}
\def\R{\mathbb R}
\def\Z{\mathbb Z}

\def\qed{$\Box$}
\def\phi{\varphi}
\def\epsilon{\varepsilon}

\begin{document}
\title{Fractional Hamiltonian Monodromy from a Gauss-Manin Monodromy}
\author{D. Sugny$^1$,P. Marde\v{s}i\'c$^2$,M. Pelletier$^2$,A.
Jebrane$^2$,H. R. Jauslin$^1$}
\address{$^1$ Institut Carnot de
Bourgogne, UMR 5209 CNRS-Universit\'e de Bourgogne, BP 47870,
21078 Dijon, France}
\address{$^2$ Institut de Math\'ematiques de Bourgogne, UMR CNRS 5584, BP
47870, 21078 Dijon, France} \email{dominique.sugny@u-bourgogne.fr}
\date{\today}
\maketitle
\begin{abstract}
Fractional Hamiltonian Monodromy is a generalization of the notion
of Hamiltonian Monodromy, recently introduced by N. N.
Nekhoroshev, D. A. Sadovski\'i and B. I. Zhilinski\'i for
energy-momentum maps whose image has a particular type of
non-isolated singularities. In this paper, we analyze the notion
of Fractional Hamiltonian Monodromy in terms of the Gauss-Manin
Monodromy of a Riemann surface constructed from the
energy-momentum map and associated to a loop in complex space
which bypasses the line of singularities. We also prove some
propositions on Fractional Hamiltonian Monodromy for $1:-n$ and
$m:-n$ resonant systems.
\end{abstract}
\section{Introduction}
We consider an integrable system on a four dimensional symplectic
manifold defined by an energy-momentum map. For a proper map, the
Liouville-Arnold theorem allows to foliate the phase space by tori
or a disjoint union of tori over the regular values of the image
of the map. Hamiltonian monodromy is the monodromy of this
fibration \cite{duist,cushman}. The word \textit{Hamiltonian} is
added to distinguish this monodromy from the Gauss-Manin monodromy
of Riemann surfaces which is also used in this paper. A non
trivial monodromy can be expected if the set of regular values of
the image of the energy-momentum map is not simply connected.
Hamiltonian monodromy has profound implications both in classical
and quantum mechanics \cite{san1} since it is the simplest
topological obstruction to the existence of global action-angle
variables \cite{duist} and thus of \textit{global good quantum
numbers} \cite{san1}. The phenomenon of Hamiltonian monodromy has
been exhibited in a large variety of physical systems both in
classical and quantum mechanics
\cite{ezra,sadov,kozin,child2,sadov2,waalkens,co22,efstathiou}.

The presence of non-trivial monodromy in energy-momentum maps with
isolated singularities of focus-focus type is now
well-established. The non-trivial monodromy in the spherical
pendulum is due to this singularity \cite{duist,cushman}.
Recently, the definition of Hamiltonian monodromy has been
extended to characterize not only isolated singularities but also
some types of non-isolated singularities, leading to the concept
of Fractional Hamiltonian Monodromy
\cite{frac1,frac2,frac3,frac4}. More precisely, one considers an
energy-momentum map with a 1-dimensional set $C$ of \textit{weak
critical values} defined by the property that each point of this
set lifts to a particular type of singular torus, a curled torus,
i.e. for the simplest case two cylinders glued together along a
line whose extremities are identified after a half-twist. The
standard Hamiltonian monodromy describes the possible
non-triviality of a 2-torus bundle over a loop in the set of
regular values of the image of the energy-momentum map. The
monodromy matrix is an automorphism of the first homology group
$H_1$ of the torus with integer coefficients. Fractional monodromy
can appear if the set of admissible paths is enlarged to include
loops which cross the singular line $C$. For such paths, the
singularity of $C$ being sufficiently weak, it can be shown that
the monodromy action can still be defined but only on a subgroup
of $H_1$. The formal extension of this action to the whole group
leads to monodromy matrices with fractional coefficients and to
the denomination fractional monodromy. One of the main motivations
for the introduction of this new concept is given by the quantum
manifestation of monodromy in the discrete joint spectrum of the
energy-momentum map \cite{san1,frac2}. This spectrum can be
represented as a lattice of points in $\R^2$. The focus-focus
singularity can be detected by a point defect of this lattice
which prevents it to be a regular lattice isomorphic to $\Z^2$. In
the same way, fractional monodromy can be interpreted as a line
defect of the lattice and appears therefore as a natural
generalization of standard hamiltonian monodromy.

The presence of fractional monodromy has been shown in a system of
coupled oscillators in $m:-n$ resonance with $m$ or $n$ different
from 1. Two constructions have been given based on geometric
\cite{frac2,nekonew} or analytic \cite{frac4} arguments to define
rigorously the crossing of $C$. The geometric construction
consists in following a basis of cycles of a regular torus through
the crossing of $C$. Not all the cycles can cross continuously the
singularity, only those corresponding to a subgroup of $H_1$ can.
When $C$ is crossed, one allows cycles to break up and reconnect,
the orientation of the cycles being preserved. The second
construction uses, as in the original paper of Duistermaat
\cite{duist,cushman}, the period lattice of the torus
\cite{arnold} which is however not defined on the singular line
$C$. The period lattice is defined through two functions $\Theta$
and $T$ at each point of the regular values of
 the image of the energy-momentum map. Some regularizations of
these functions can be made in order to cross continuously the
line of singularities \cite{frac4}. Note that the preceding
geometric point of view can be reconstructed from this analytic
approach since the basis of cycles can be determined from the
functions $\Theta$ and $T$.

The goal of this work is to study fractional monodromy by
complexifying the phase space in order to bypass the line of
singularities. Somewhat similar studies for the Lagrange top and
the spherical pendulum have already been published
\cite{compl1,compl2,compl3} and have highlighted the relation
between Hamiltonian and complex monodromy. A parallel can also be
made with Bohr-Sommerfeld rules for semi-classical quantization.
Such calculations can be undertaken in the $C^\infty$
\cite{colin1,colin2,colin3} or in the analytic context
\cite{voros,pham1,pham2}. The real approach needs regularization
of the sub-principal term whereas the complex approach avoids such
problems by avoiding the singularity. In this paper, we show that
fractional hamiltonian monodromy is given by a Gauss-Manin
monodromy \cite{AGZV,zoladek} of a Riemann surface constructed
from the energy-momentum map. The construction can be made in the
reduced phase space \cite{cushman} which is well suited to the
introduction of Riemann surfaces. The Gauss-Manin connection is
defined for a complex semi-circle around the line $C$ of
singularities and can be calculated by applying the
Picard-Lefschetz theory. We also introduce the complex extension
of the functions $\Theta$ and $T$ which are viewed as integrals of
rational forms over a cycle of the Riemann surface. In other
words, the regularizations of $\Theta$ and $T$ in the real
approach are replaced by a complex continuation of these
functions. The variations of $\Theta$ and $T$ along the bypass in
the complex domain are deduced from the Gauss-Manin monodromy. The
functions $\Theta$ and $T$ allow us to go back to the real
approach and to define a real monodromy.
 We show that this monodromy corresponds to
  fractional hamiltonian monodromy asymptotically in the limit
 where the radius of the complex semi-circle goes to 0. Using this construction,
we recover the results of the real approach obtained in Refs.
\cite{frac1,frac2,frac3,frac4} for the 1:-2 resonance and Ref.
\cite{nekonew} for $m:-n$ resonance. Moreover, for $1:-n$ and
$m:-n$ resonant systems, we give new proofs of these results from
the real and the complex approaches. A
 geometric point of view of the Gauss-Manin monodromy can be
given by inspecting the motion of the branching points of the
Riemann surface along the semi-circle around $C$. For $1:-n$
resonant systems, the Riemann surface has locally $n$ complex
branching points in a neighborhood of $C$ lying on a circle around
the origin with an angle of $2\pi/n$ between each other. Along a
complex semi-circle around $C$, the $n$ points turn by an angle of
$2\pi/n$ and exchange their positions. The variation of $\Theta$
near $C$ is then calculated as a residue of a given 1-form. This
characterizes the line of singularities $C$ and fractional
monodromy.

The organization of this article is as follows : We first consider
the real approach and we determine the monodromy matrices for
1:-2, $1:-n$ and $m:-n$ resonances in Sec. \ref{real}. We next
show in Sec.
 \ref{complex} how fractional monodromy can be defined in the
 complex approach and we detail its geometric interpretation. We recover the different
 results obtained in the real approach. Concluding remarks and perspectives are given in
 Sec. \ref{conc}. Appendices \ref{appgeom} and \ref{semi} present a schematic
 representation of the geometric construction of fractional monodromy for
 1:-2 and $1:-n$ resonant systems in the real approach and the semi-classical
 point of view for $1:-n$ and $m:-n$ resonances.
 The material of these appendices complements the
 existing literature on these two points. Appendix \ref{appred} finally
 deals with the reduction procedure in the complex approach.
\section{The real approach}\label{real}
The goal of this section is to furnish a short overview of
fractional hamiltonian monodromy in the real approach. Starting
from the example introduced in Ref. \cite{frac3,frac4} for the
1:-2 resonance, we extend it to $1:-n$ and $m:-n$ resonances. The
core of the results presented in this section are already
contained in \cite{frac1,frac2,frac3,frac4}. We describe it in
some detail since we need the results for the complex approach.
Furthermore, we present analytic computations in this section and
a geometric construction in appendix \ref{appgeom} which are
slightly different from the original ones and give a new view on
fractional hamiltonian monodromy. The originality of the
analytical computation for the 1:-2 resonance relies in the
introduction of a local description of the energy-momentum map
near the origin of the bifurcation diagram which considerably
simplifies the computation of the monodromy matrix. Since these
arguments are generalizable to $1:-n$ and $m:-n$ resonances, they
allow us to prove some propositions stated in
\cite{frac2,frac3,nekonew}. For the geometric description, we
introduce the standard representation of a torus i.e. a rectangle
whose some edges are identified. This geometric construction can
be extended straightforwardly to $1:-n$ resonant systems.
\subsection{The 1:-2 resonance} \label{real1m2}
We consider the symplectic manifold $M=T^*\R^2$ with standard
symplectic form $\omega=dq_1\wedge dp_1+dq_2\wedge dp_2$. We
introduce the energy-momentum map $F=(H,J):M\to\R^2$ where the two
 functions $J$ and $H$ have zero Poisson brackets
$\{J,H\}=0$. Following Refs. \cite{frac1,frac2,frac3,frac4}, we
choose a system corresponding to the $1:-2$ resonance defined by
\begin{eqnarray} \label{eq1}
\left\{ \begin{array}{ll} H=\sqrt{2}[\left(q_1^2-p_1^2\right)p_2+
2q_1p_1q_2]+2\varepsilon (q_1^2+p_1^2)(q_2^2+p_2^2)\\
J=\frac{1}{2}\left(q_{1}^2+p_{1}^2\right)-\left(q_2^2+p_2^2\right)
\end{array} \right. \ ,
\end{eqnarray}
where $\varepsilon$ is a non-zero real number. We denote by
$\mathcal{R}$ the image of $F$ and by $\mathcal{R}_{reg}$ the
regular values of $\mathcal{R}$. We recall that a point
$M\in\mathcal{R}$ is regular if the 1-forms $dH$ and $dJ$ are
linearly independent at all points of $F^{-1}(h,j)$.

The flow of $J$ defines an $S^1$-action on the phase space but
this action is not principal since the isotropy groups of the
points $\{p_1=0,q_1=0,p_2,q_2\}$ are isomorphic to $\Z/2\Z$. The
reduced phase space $J^{-1}(j)/S^1$ can be constructed by using
the algebra of invariant polynomials with values in $\R$ which is
generated by \cite{frac3,cushman} :
\begin{eqnarray} \label{eq2}
\left\{ \begin{array}{llll}
J({\bf p,q})=1/2\left(q_{1}^2+p_{1}^2\right)-\left(q_2^2+p_2^2\right)\\
\pi_1({\bf p,q})=1/2\left(q_{1}^2+p_{1}^2\right)+\left(q_2^2+p_2^2\right)\\
\pi_2({\bf p,q})=\sqrt{2}[\left(q_1^2-p_1^2\right)q_2-2q_1p_1p_2]\\
\pi_3({\bf
p,q})=\sqrt{2}[\left(q_1^2-p_1^2\right)p_2+2q_1p_1q_2]\end{array}\right.
 \ ,
\end{eqnarray}
with the constraint $|J|\leq \pi_1$. The reduced phase spaces
$P_j=J^{-1}(j)/S^1$ are defined in the space
$\R^3=\left(\pi_1,\pi_2,\pi_3\right)$ by the equations :
\begin{equation} \label{eq4}
\pi_{2}^2+\pi_{3}^2=\left(\pi_{1}-j\right)\left(\pi_{1}+j\right)^2
\ ,
\end{equation}
and correspond to non-compact surfaces with a conical singularity
for $j<0$. Having introduced the invariant polynomials, some
comments can be made on the choice of $H$. Since $\{H,J\}=0$, it
can be shown that if $H$ is polynomial in $\left(p_i,q_i\right)$
then it can be written as a polynomial function in $J$, $\pi_1$,
$\pi_2$ and $\pi_3$. We also notice that the reduced phase space
for the 1:-2 resonance being non-compact, not every choice of $H$
leads to a proper map for the energy-momentum map $F=(H,J)$. This
explains why a term of degree four has been added to $H$ [see Eq.
(\ref{eq1})]. The image of the energy-momentum map defined by Eq.
(\ref{eq1}) has the particularity to present a line of
singularities $C$, each point of this line except the origin lifts
to a singular curled torus (the origin lifts to a pinched-curled
torus). The topology of these singular tori can be determined by
the intersection of the reduced phase space $P_j$ with the level
set $H_j=h$ (see Sec. \ref{res1n}). $H_j$ is the reduced
Hamiltonian, i.e., a map from $P_j$ to $\R$ that sends a point of
$P_j$ to $H(\pi_1,\pi_2,\pi_3,j)$. Moreover, since the two
energy-momentum maps $F=(H,J)$ and $F'=(H-f(J),J)$, where $f$ is a
polynomial function, define up to a diffeomorphism the same
fibration of the phase space, we can consider an example such that
$H=0$ for the line of singularities $C$. Figure \ref{fig1}
displays the bifurcation diagram of $F$ (see \cite{frac3,frac4}
for details on this construction).
\begin{figure}
\begin{center}
\includegraphics[scale=0.4]{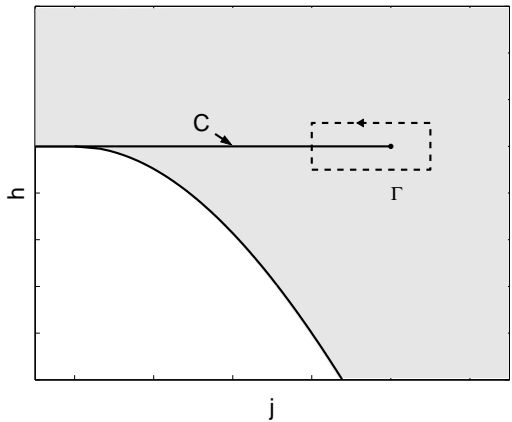}
\caption{\label{fig1} Image $\mathcal{R}$ of the energy-momentum
map $F$ (in grey). The singular line $C$ is represented by the
horizontal solid line. The small full dot indicates the position
of the image of the pinched-curled torus. A loop $\Gamma$ crossing
transversally $C$ is depicted in dashed lines.}
\end{center}
\end{figure}
Singular points (i.e. not regular) are represented by solid lines
in Fig. \ref{fig1}. $C$ is a line of weak singularities. The word
\textit{weak} means that for each point of $C$, there exist points
of the corresponding pre-image such that the rank of $F$ is 1. We
recall that for a regular point, the rank of $F$ is always equal
to 2 and that the rank of $F$ is 0 for one point of $F^{-1}(0,0)$.
\subsection{Fractional Hamiltonian Monodromy} \label{fhm}
We begin this section by recalling some basic facts about integer
and fractional monodromy. The word \textit{hamiltonian} in
hamiltonian monodromy will be omitted when confusion is unlikely
to occur.

Standard or integer monodromy is defined for a loop $\Gamma$ along
regular values of $\mathcal{R}$. Fractional monodromy is obtained
by extending the possible loops, allowing them to cross some
particular type of weak singular lines such as $C$. In the
standard case, we consider a point $(h,j)$ of $\Gamma$ which lifts
to a torus $T^2(h,j)$. We fix a basis of the homology group
$H_1(T^2(h,j),\Z)$. There exists a natural connection of the torus
bundle which allows to transport this basis along $\Gamma$
\cite{cushman,bates}. The monodromy matrix, which is an
automorphism of $H_1$, is the holonomy of this connection. This
construction can be generalized to fractional monodromy but only a
subgroup of $H_1(T^2(h,j),\Z)$ can be transported continuously
across the line $C$ \cite{frac2,frac4}. These remarks can be
understood by the following construction. Let $(h,j)\in
\mathcal{R}_{reg}$. Denoting by $\phi_J$ and $\phi_H$ the flows
associated to the Hamiltonians $J$ and $H$, the period lattice of
$F$ at a point $(h,j)$ is the set
\begin{equation} \label{fhm1}
\{ (t_1,t_2)\in \R^2 | \phi_J^{t_1}\circ\phi_H^{t_2}(z)=z\} \ ,
\end{equation}
for all $z\in F^{-1}(h,j)$. A basis for this period lattice, which
is isomorphic to $\Z^2$, is given by the vectors $v_1=(2\pi,0)$
and $v_2=(-\Theta,T)$ where $\Theta$ is the rotation angle and $T$
the first return time of the flow $\phi_H$ defined as follows
\cite{duist,cushman}. The Hamiltonian $J$ generates an
$S^1$-action on $F^{-1}(h,j)$. We denote by $\theta$ an angle
conjugated to the action $J$. Following $\phi_J$ which is
parameterized by $\theta$, one goes back to the starting point
when $\theta$ increases by $2\pi$, which gives $v_1$. If we
consider now $\phi_H$ from a point of an orbit of the flow
$\phi_J$, one sees that the first intersection of these two flows
takes places at time $T$. The two points of intersection of the
two flows define two angles $\theta_f$ and $\theta_i$ and the
twist $\Theta=\theta_f-\theta_i$ which is determined with respect
to the direction of $\phi_J$. Note that a different choice of the
angle $\theta$ leads to a different basis for the period lattice.
The corresponding rotation angles $\Theta$ differ by a multiple of
$2\pi$. The monodromy matrix associated to a loop lying in the
regular values of the image of the energy-momentum map is related
to the behavior of the functions $\Theta$ and $T$ along this loop.
For the standard monodromy, after a counterclockwise loop around
an isolated critical value (focus-focus singularity) it can be
shown that the rotation angle is increased by $2\pi$ whereas the
first return time is unchanged. As a consequence, $v_1$ is
transformed into $v_1$ and $v_2$ into $-v_1+v_2$, and thus the
monodromy matrix $M$ written in the local basis $(v_1,v_2)$ is
equal to
\begin{eqnarray} \label{fhm2}
M= \left( \begin{array}{cc}
1 & 0 \\
-1 & 1
\end{array} \right) \ .
\end{eqnarray}
The two functions $\Theta$ and $T$ allow to define a basis of
cycles for the homology group $H_1(F^{-1}(h,j),\Z)$ and thus to
recover a more geometric point of view. A basis
$([\beta_1],[\beta_2])$ of $H_1(F^{-1}(h,j),\Z)$ is given by the
cycles associated respectively to the flows of the vector fields
\begin{eqnarray}\label{fhm2a}
\left\{ \begin{array}{ll}
X_1=2\pi X_J \\
X_2=-\Theta(h,j)X_J+T(h,j)X_H
\end{array} \right. \ .
\end{eqnarray}
The flows $\phi_{X_1}^t$ ($t\in [0,1]$) and $\phi_{X_2}^t$ ($t\in
[0,1]$) generate respectively the closed cycles $\beta_1$ and
$\beta_2$. These two cycles are schematically represented in Fig.
\ref{fig1a} for $\Theta=-\pi$. In the basis
$([\beta_1],[\beta_2])$, the monodromy matrix is given by the same
matrix as Eq. (\ref{fhm2}) obtained for the period lattice.
\begin{figure}
\begin{center}
\includegraphics[scale=0.4]{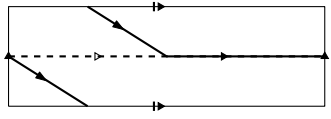}
\caption{\label{fig1a} Schematic representation on the torus
$F^{-1}(h,j)$ of the cycles $\beta_1$ and
 $\beta_2$ depicted respectively in dashed and solid lines.}
\end{center}
\end{figure}

The situation is slightly more complicated for fractional
monodromy. The analytic construction of fractional monodromy
follows the same steps as for the standard case \cite{frac4}.
Returning to the example of Eq. (\ref{eq1}) and considering the
loop $\Gamma$ of Fig. \ref{fig1}, the question which naturally
arises is the definition of the crossing of the line $C$ since
$\Theta$ has a discontinuity of size $\pi$ on this line and $T$
diverges. The idea of the method proposed in Ref. \cite{frac4}
consists in prolonging by continuity the function $\Theta$ (the
continuous function is called $\tilde{\Theta}$) and in rescaling
the time to obtain a finite first return time denoted $\tau$. Note
that this rescaling does not modify the definition of the cycles
but only the time to cover them. This leaves therefore the
monodromy matrix unchanged. $\tilde{\Theta}$ is defined by
$\tilde{\Theta}(h,j)=\Theta(h,j)$ for points $(h,j)$ of $\Gamma$
before the crossing of the line $C$ and by
$\tilde{\Theta}(h,j)=\Theta(h,j)+\pi$ for points after
\cite{frac4}. A basis of the period lattice is given by the
vectors $v_1=(2\pi,0)$ and $v_2=(-\tilde{\Theta},\tau)$ and the
cycles $\beta_1$ and $\beta_2$ are now associated to the vector
fields
\begin{eqnarray}\label{fhm2b}
\left\{ \begin{array}{ll}
X_1=2\pi X_J \\
\tilde{X}_2=-\tilde{\Theta}(h,j)X_J+\tau(h,j)X_H
\end{array} \right. \ .
\end{eqnarray}
The construction of these cycles has however to be carefully
examined. $\phi_{X_1}^t$ ($t\in[0,1]$) generates a cycle $\beta_1$
for all points of $\Gamma$. Before the crossing,
$\phi_{\tilde{X}_2}^t$ ($t\in[0,1]$) generates a closed cycle
$\beta_2$ but after the crossing, $\phi_{\tilde{X}_2}^t$
($t\in[0,1]$) generates only half of a cycle. To get a complete
cycle, we thus have to take $t\in[0,2]$. This means geometrically
that only a cycle $\beta_2$ covered twice can be transported
continuously across the line $C$. Thus, not all cycles in the
homology group can be transported along $\Gamma$. Only a subgroup
corresponding to the cycles that are run twice by $\tilde{X}_2$
can be transported. In terms of the period lattice, the crossing
of $C$ is only possible for the sublattice generated by $v_1$ and
$2v_2$. We can then define the monodromy matrix for a
counterclockwise loop crossing transversally $C$ once. The
monodromy matrix $M$ reads in the basis $(v_1,2v_2)$ or in the
basis $([\beta_1],[2\beta_2])$
\begin{eqnarray} \label{fhm3}
M= \left( \begin{array}{cc}
1 & 0 \\
-1 & 1
\end{array} \right) \ .
\end{eqnarray}
Extending formally the definition of $M$ to the whole homology
group or the whole period lattice, we obtain in the basis
$(v_1,v_2)$ or in the basis $([\beta_1],[\beta_2])$
\begin{eqnarray} \label{fhm3}
M= \left( \begin{array}{cc}
1 & 0 \\
-\frac{1}{2} & 1
\end{array} \right) \ .
\end{eqnarray}
We finally note that this \textit{generalized} monodromy is still
topological in the sense that it depends only on the homotopy
equivalence class of the loop considered. We also point out that
these cycles prolonged continuously to the curled torus allow to
recover the geometric construction of fractional monodromy
described in appendix \ref{appgeom}.
\subsection{Local computation of the monodromy matrix}
The determination of the monodromy matrix is based on the behavior
of the functions $\Theta$ and $\tau$ on the line of singularities
$C$. More precisely, the monodromy matrix can be constructed
uniquely from the size of the discontinuity of $\Theta$ and from
the fact that $\tau$ is continuous. In Ref. \cite{frac4}, the
computation was done by using global expressions for $\Theta$ and
$\tau$ in terms of elliptic integrals. It is clear that such a
global calculation can be expected to be done explicitly only for
simple energy-momentum maps.

We propose a computation of the fractional monodromy matrix based
on local arguments for the energy-momentum map $F$, i.e., a
Puiseux expansion in $h$ and $j$ around $(h=0,j=0)$. This
expansion is not trivial because a particular dissymmetry in $h$
and $j$ has to be preserved.
\begin{lemma} \label{lemma1}
In the variables $(\pi_1,\pi_2,\pi_3,J)$ and for a point
$(h,j)\in\mathcal{R}_{reg}$, the functions $\Theta$ and $\tau$
 are given by the following expressions :
\begin{equation} \label{fhm4}
\Theta(h,j)=h\int_{\pi_1^-}^{\pi_1^+}\frac{d\pi_1}{(j+\pi_1)\sqrt{Q(\pi_1)}}
\ ,
\end{equation}
and
\begin{equation} \label{fhm5}
\tau(h,j)=\frac{1}{2}\int_{\pi_1^-}^{\pi_1^+}\frac{j+\pi_1}{\sqrt{Q(\pi_1)}}d\pi_1
\ ,
\end{equation}
where $Q(\pi_1)$ is a polynomial given by
\begin{equation} \label{fhm6}
Q(\pi_1)=(\pi_1-j)(\pi_1+j)^2-[h-\epsilon(\pi_1^2-j^2)]^2 \ .
\end{equation}
$\pi_1^+$ and $\pi_1^-$ are the two largest real roots of $Q$ with
$\pi_1^-<\pi_1^+$.
\end{lemma}
\begin{proof}
see \cite{frac4} and Sec. (\ref{mnres}) for a more general
construction. We recall that $\tau$ is the first return time of
the flow of the rescaled vector field
$\frac{1}{q_1^2+p_1^2}X_H$.\qed
\end{proof}
\begin{proposition}
The monodromy matrix for a counterclockwise oriented loop $\Gamma$
crossing
 once the line $C$ transversally at a point different from the
origin (see Fig. \ref{fig1}) is given by
\begin{eqnarray} \label{fhm7a}
M= \left( \begin{array}{cc}
1 & 0 \\
-1/2 & 1
\end{array} \right) \ .
\end{eqnarray}
\end{proposition}
\begin{proof}
The monodromy matrix is given by the behavior of $\Theta$ and
$\tau$ in the neighborhood of the line $C$. Taking $j<0$ fixed and
finite, the limits $\lim_{h\to 0^\pm}\Theta(h,j)$ and $\lim_{h\to
0^\pm}\tau(h,j)$ have been calculated in \cite{frac4}. Elliptic
integrals and asymptotic expansions of these integrals were used.

We propose a simpler computation by considering the asymptotic
limit $j\to 0$. For that purpose, we analyze the roots of the
polynomial $Q$ as $h$ and $j$ go to zero. Since a qualitative
change of the functions $\Theta$ and $\tau$ is expected when the
polynomial $Q$ has a multiple complex root, we determine the
complex discriminant locus of $Q$ near the origin $h=j=0$. This
point will be made clearer with the introduction of Riemann
surfaces in Sec. \ref{complex} but here it gives the way in which
the two limits $h\to 0$ and $j\to 0$ should be taken. The
discriminant locus in the real approach has already been
calculated since it corresponds to the line of singularities of
the bifurcation diagram, i.e., to the points where the 1-forms
$dH$ and $dJ$ are linearly dependent (see Fig. \ref{fig1}). We
introduce the variable $x=j+\pi_1$ and $h$, $j$ and $x$ are taken
complex. Three roots of $Q$ vanish for $h=j=0$.

Let us assume that $x$, $j$ and $h$ go to zero. Constructing the
Newton polyhedron associated to $Q$ \cite{kirwan}, we obtain the
principal part $Q_N$ of $Q$ which can be written as
\begin{equation}\label{fhm7}
Q_N=x^3-2jx^2-h^2 \ .
\end{equation}
We notice that this principal part is symmetric with respect to
$h$ which is not the case for the polynomial $Q$. Simple algebra
leads to the following asymptotic discriminant locus
\begin{eqnarray}\label{fhm8}
\left\{ \begin{array}{ll}
h=\pm\sqrt{j^3(-\frac{32}{27})} \\
h=0
\end{array} \right. \ ,
\end{eqnarray}
denoted $\Delta$. This complex locus is displayed for $j\in\R$ in
Fig. \ref{fig6}. Equation (\ref{fhm7}) also shows that the weights
2, 2 and 3 can respectively be attributed to $x$, $j$ and $h$. In
other words, if we introduce the small parameter $r$, the
principal parts of $x$, $j$ and $h$ can be written
\begin{eqnarray}\label{fhm8a}
\left\{ \begin{array}{lll}
x_N=\bar{x}r^2 \\
j_N=\bar{j}r^2 \\
h_N=\bar{h}r^3
\end{array} \right. \ .
\end{eqnarray}
The monodromy can be computed along a small loop in $\mathcal{R}$
around the origin. This loop can be parameterized by
$(\bar{h},\bar{j})\in S^1$ which allows to derive a local version
of the computation of fractional monodromy. We thus let
$\bar{h}\to 0$ while keeping $\bar{j}$ fixed and finite. Note that
it is equivalent to consider the limits $h,j\to 0$ with the
condition $h=o(j^{3/2})$, since on $\Delta$ we have asymptotically
$h=O(j^{3/2})$. In order to determine the leading terms of the
roots of $Q$ in this case, we construct the Newton polygon of
$Q_N$. The approximate solutions fulfill
\begin{equation}\label{fhm9}
-2\bar{j}\bar{x}^2-\bar{h}^2=0 \ .
\end{equation}
Since the sum of the three roots of $Q$ that go to zero when
$h,j\to 0$ is $2j$, and the sum of the four roots is
$\frac{1}{\varepsilon^2}$, one deduces that the principal parts of
the roots are given in the original variables by
\begin{eqnarray}\label{fhm10}
\left\{ \begin{array}{llll}
x_1=2j \\
x_2=\frac{-h}{\sqrt{-2j}} \\
x_3=\frac{h}{\sqrt{-2j}} \\
x_4=\frac{1}{\varepsilon^2}
\end{array} \right. \ .
\end{eqnarray}
These expressions can be compared with the ones given in Ref.
\cite{frac4} where $j<0$ is fixed and $h\to 0$ :
\begin{eqnarray}\label{fhm10a}
\left\{ \begin{array}{llll}
x_1=2j \\
x_2=\frac{-h}{\sqrt{-2j}+2\varepsilon j} \\
x_3=\frac{h}{\sqrt{-2j}-2\varepsilon j} \\
x_4=2j+\frac{1}{\varepsilon^2}+\frac{2\varepsilon
h}{1+2\varepsilon^2j}
\end{array} \right. \ .
\end{eqnarray}
We now calculate $\lim_{\bar{h}\to 0^\pm,\bar{j}<0}\Theta(h,j)$
and $\lim_{\bar{h}\to 0,\bar{j}<0}\tau(h,j)$. We consider first
$h>0$ and $j<0$. $\Theta(h,j)$ can be written as
\begin{equation} \label{fhm11}
\Theta(h,j)=\frac{h}{i\varepsilon}\int_{x_3}^{x_4}
\frac{dx}{x\sqrt{(x-x_1)(x-x_2)(x-x_3)(x-x_4)}} \ .
\end{equation}
We determine only the principal term of the asymptotic expansion
of $\Theta$. The symbol $\sim$ represents the equivalence in the
limit $h\to 0$, $j\to 0$ and $h=o(j^{3/2})$. From Eqs.
(\ref{fhm10}), we obtain
\begin{equation} \label{fhm12}
\Theta(h,j)\sim\frac{h}{i\varepsilon}\int_{h/\sqrt{-2j}}^{1/\varepsilon^2}
\frac{dx}{x\sqrt{(x-2j)(x+h/\sqrt{-2j})(x-h/\sqrt{-2j})(x-1/\varepsilon^2)}}
\ .
\end{equation}
We decompose the preceding integral into three integrals by
introducing the terms $k$ and $k'$ which go to 0 such that
$\frac{|h|}{\sqrt{-2j}}\ll k\ll -2j \ll k'\ll 1$. $k$ and $k'$ are
chosen for instance as $j^\alpha$. The notation $a\ll b$ means
that the ratio $a/b\to 0$ as $a$ and $b$ go to 0. The three
integrals are taken over the intervals $[|h|/\sqrt{-2j},k]$,
$[k,k']$ and $[k',1/\varepsilon^2]$. It can be shown that the
limit of the last two integrals is zero. The first integral reads
\begin{equation} \label{fhm13}
\Theta(h,j)\sim\frac{h}{i\varepsilon}\int_{h/\sqrt{-2j}}^{k}
\frac{dx}{x\sqrt{(-2j)(x^2-h^2/(-2j))(-1/\varepsilon^2)}} \ ,
\end{equation}
which can be rewritten as
\begin{equation} \label{fhm14}
\Theta(h,j)\sim h\int_{h/\sqrt{-2j}}^{k}
\frac{dx}{x\sqrt{(-2j)(x^2-h^2/(-2j))}} \ .
\end{equation}
The change of variables $x=\frac{h}{\sqrt{-2j}}u$ leads to
\begin{equation} \label{fhm15}
\Theta(h,j)\sim \int_{1}^{+\infty} \frac{du}{u\sqrt{u^2-1}} \ .
\end{equation}
Using the fact that
\begin{equation} \label{fhm15}
\int\frac{du}{u\sqrt{u^2-1}}=\arctan[\sqrt{u^2-1}] \ ,
\end{equation}
one finally arrives to
\begin{equation} \label{fhm16}
\lim_{h\to 0^+, j\to 0^-, h=o(j^{3/2})}\Theta(h,j)=\frac{\pi}{2} \
.
\end{equation}
Similar calculations for $h<0$ give
\begin{equation} \label{fhm17}
\lim_{h\to 0^-, j\to 0^-, h=o(j^{3/2})}\Theta(h,j)=\frac{-\pi}{2}
\ .
\end{equation}
We then deduce that the discontinuity of $\Theta$ is equal to
$\pi$.

$\tau$ can be calculated along the same lines. This time the first
two terms go to zero and we only determine the last one
\begin{equation} \label{fhm18}
\tau(h,j)\sim
\frac{1}{2i\varepsilon}\int_{k'}^{1/\varepsilon^2}\frac{xdx}{\sqrt{
(x-x_1)(x-x_2)(x-x_3)(x-x_4)}} \ .
\end{equation}
Simple algebra leads to
\begin{equation} \label{fhm19}
\tau(h,j)\sim
\frac{1}{2i\varepsilon}\int_{k'}^{1/\varepsilon^2}\frac{dx}{\sqrt{
x(x-1/\varepsilon^2)}} \ ,
\end{equation}
and using the fact that
\begin{equation} \label{fhm20}
\int \frac{dx}{\sqrt{x(x-1/\varepsilon^2)}}=2\ln
[\sqrt{x}+\sqrt{x-1/\varepsilon^2}] \ ,
\end{equation}
one obtains that
\begin{equation} \label{fhm17}
\lim_{h\to 0^\pm, j\to 0^-,
h=o(j^{3/2})}\tau(h,j)=\frac{-\pi}{2\varepsilon} \ .
\end{equation}
$\tau$ is therefore continuous on the line $C$. The monodromy
matrix is finally deduced from the behavior of $\Theta$ and $\tau$
in the neighborhood of $C$. We follow for that purpose the construction
of Ref. \cite{frac4} which is briefly recalled in Sec. \ref{fhm}.\qed
\end{proof}
\begin{remark}
The preceding computation being local does not show the
topological character of fractional monodromy i.e. its
independence with respect to the homotopically equivalent
 loops considered or more simply with respect to $j$. This point has been proved
 in the real approach in Ref. \cite{frac4} and will be proved in the complex approach in Sec.
 \ref{complex}.
\end{remark}
\subsection{$1:-n$ resonant system}\label{res1n}
We consider the energy-momentum map $F=(H,J)$ where $J$ is given
by
\begin{equation} \label{res1n1}
J=\frac{1}{2}[(p_1^2+q_1^2)-n(q_2^2+p_2^2)] \ ,
\end{equation}
with $n\geq 2$. We assume that the bifurcation diagram of $F$ is
locally in the neighborhood of the origin given by Fig.
\ref{fig2}. The singular locus corresponds to ($h=0,j\leq 0$).
Each point of this locus lifts to a singular torus i.e. a
pinched-curled torus for the origin and an $n-$curled torus for
the other points. A $k$-curled torus is a singular torus for which
one cycle is covered k-times while the others only once.
\begin{figure}
\begin{center}
\includegraphics[scale=0.4]{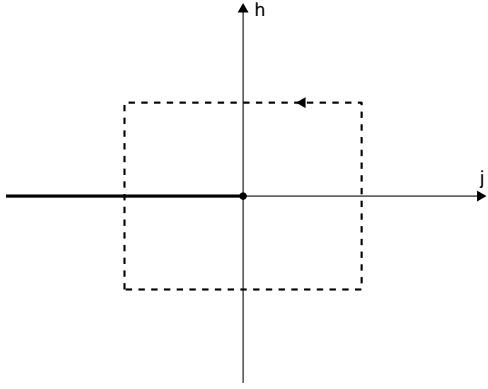}
\caption{\label{fig2} Local bifurcation diagram in the
neighborhood of the origin $(h=0,j=0)$ for the $1:-n$ resonant
system. The singular locus is represented by the large solid line.
The small full dot indicates the position of the origin. The
dashed line depicts a loop used to calculate the fractional
monodromy matrix.}
\end{center}
\end{figure}
We now restrict the discussion to a particular family of
energy-momentum maps having locally the bifurcation diagram of
Fig. \ref{fig2}. As in Sec. \ref{fhm}, the first step consists in
determining the invariant polynomials associated to the momentum
$J$. We have \cite{frac3}
\begin{eqnarray} \label{res1n3}
\left\{ \begin{array}{llll}
J({\bf p,q})=1/2[m\left(q_{1}^2+p_{1}^2\right)-n\left(q_2^2+p_2^2\right)]\\
\pi_1({\bf p,q})=1/2[m\left(q_{1}^2+p_{1}^2\right)+n\left(q_2^2+p_2^2\right)]\\
\pi_2({\bf p,q})=\sqrt{n^mm^n}\Re[(q_1+ip_1)^n(q_2+ip_2)^m]\\
\pi_3({\bf
p,q})=\sqrt{n^mm^n}\Im[(q_1+ip_1)^n(q_2+ip_2)^m]\end{array}\right.
 \ ,
\end{eqnarray}
for $m\geq 1$ and $n\geq 1$. The reduced phase space $P_j$ is
defined by
\begin{equation} \label{res1n3a}
\pi_2^2+\pi_3^2=(\pi_1+j)^n(\pi_1-j)^m \ ,
\end{equation}
with the condition $\pi_1\geq|j|$.
\begin{definition}\label{def1}
We consider the set $\mathcal{F}$ of energy-momentum maps
$F=(J,H)$ which can be written as
\begin{eqnarray} \label{res1n4}
F=\left\{ \begin{array}{ll}
J\\
H=\pi_3+R(\pi_1,J) \\
\end{array}\right.
\ ,
\end{eqnarray}
where $R$ is a polynomial so that $F$ is a proper map. $\pi_1$,
$\pi_3$ and $J$ are given by Eqs. (\ref{res1n3}) for $m=1$.
\end{definition}
Note that we do not search to determine or characterize the set
$\mathcal{F}$. Only some properties of the elements of
$\mathcal{F}$ will be sufficient to compute the monodromy matrix.
Simple examples can be exhibited to show that $\mathcal{F}$ is not
empty. For instance for the resonance 1:-3, we can choose
\begin{eqnarray} \label{res1n5}
F=\left\{ \begin{array}{ll}
J\\
H=\pi_3-(\pi_1-J)(\pi_1+J)^4 \\
\end{array}\right.
\ .
\end{eqnarray}
We are interested in the local behavior of $R$ near the origin or
in other words under which conditions on $R$, the image of the
corresponding energy-momentum map is given by Fig. \ref{fig2}.
\begin{lemma} \label{lemmar}
The energy-momentum map given by Eqs. (\ref{res1n4}) has locally
the bifurcation diagram of Fig. \ref{fig2} in a neighborhood of
the origin if the polynomial $R$ is of the form
\begin{equation} \label{res1n6}
R(\pi_1,J)=(\pi_1+J)^{n'}(\pi_1-J)^{m'}\tilde{R}(\pi_1,J) \ ,
\end{equation}
where $n'$ and $m'$ are positive integers such that
$n'>\frac{n}{2}$ and $n'+m'>\frac{n+1}{2}$. $\tilde{R}$ is a
polynomial in $\pi_1$ and $J$ such that the two smallest real
positive roots of the polynomial $Q$ which are larger than $j$ for
$(h,j)\in\mathcal{R}_{reg}$ are simple roots . $Q$ is the
polynomial defined by
\begin{equation} \label{res1n10}
Q=(\pi_1+j)^n(\pi_1-j)-[h-R(\pi_1,j)]^2 \ .
\end{equation}
The two roots are denoted $\pi_1^-$ and $\pi_1^+$ with
$\pi_1^-<\pi_1^+$.
\end{lemma}
\begin{proof}
The proof is based on the nature of the intersection of the
reduced phase space $P_j$ with the level sets $\{H_j=h\}$ of
equations $h=\pi_3+R(\pi_1,j)$ as $h$ and $j$ vary. Fig.
\ref{fig3} displays these intersections for three different values
of $h$, $j<0$ being fixed. The case considered in this figure is
the 1:-3 resonance and the energy-momentum map of Eqs.
(\ref{res1n5}). Fig. \ref{fig3} represents the generic topology of
the level sets $\{H_j=h\}$ which we are going to characterize. We
recall that each point of the reduced phase space lifts in the
original phase space to a circle except for the point of
coordinates $(\pi_1=-j,\pi_2=0,\pi_3=0)$ which lifts either to a
circle covered $n$ times, if $j<0$, or to a point, if $j=0$. One
deduces from the bifurcation diagram of $F$ that the intersection
of the level set $\{H_j=h\}$ with $P_j$ contains a circle passing
through the singular point $S=(\pi_1=-j,\pi_2=0,\pi_3=0)$, for
$j\leq0$ and $h=0$. We set $y=\pi_3$, $x=\pi_1+j$ and
$x'=\pi_1-j$. The local behavior near the point $S$ of the two
surfaces is given by
\begin{equation} \label{res1n7}
y=\pm x^{\frac{n}{2}}x'^{\frac{1}{2}} \ ,
\end{equation}
for $P_j$ and by
\begin{equation} \label{res1n8}
y=-R(\pi_1,j) \ ,
\end{equation}
for $\{H_j=0\}$. It is then straightforward to show that the local
behavior expected is obtained if
$R(\pi_1,j)=x^{n'}x'^{m'}\tilde{R}(\pi_1,j)$ with the conditions
$n'>\frac{n}{2}$ and $n'+m'>\frac{n+1}{2}$, $m'\geq 0$. The first
and second inequalities result respectively from the conditions
for $j<0$ and $j=0$.\qed
\end{proof}
\begin{remark}
The inequalities of lemma \ref{lemmar} are strict to ensure that a
multiplication by a constant factor of the term $R$ does not
modify the local behavior of the image of the energy-momentum map.
This point has not been assumed for the 1:-2 resonance, the
parameter $\varepsilon$ is thus chosen sufficiently small in this
case. Here, this hypothesis simplifies the computation of the
monodromy matrix as can be seen in the proofs of lemma \ref{root1}
and proposition \ref{prop1n}.
\end{remark}
\begin{figure}
\begin{center}
\includegraphics[scale=0.4]{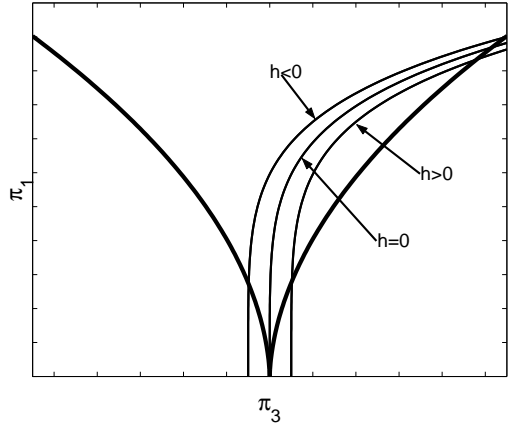}
\caption{\label{fig3} Intersections of the reduced phase space
$P_j$ (in large solid lines) with the level sets $H_j=h$ for
$h>0$, $h=0$ and $h<0$ in the plane $\pi_2=0$. The energy-momentum
map corresponding to this diagram is given by Eqs.
(\ref{res1n5}).}
\end{center}
\end{figure}
\begin{proposition} \label{prop1n} The monodromy matrix of a 1:-n resonant
system of the form (\ref{res1n4}) is given for a loop turning
counterclockwise around the origin (see Fig. \ref{fig2}) and
crossing once transversally the line $C$ by
\begin{eqnarray} \label{res1n2}
M= \left( \begin{array}{cc}
1 & 0 \\
-1/n & 1
\end{array} \right) \ .
\end{eqnarray}
\end{proposition}
\begin{remark}
Propositions \ref{prop1n} and \ref{propmn} (see Sec. \ref{mnres})
 were formulated as conjectures in Refs. \cite{frac2,frac3}. They are based on the analysis
of the quantum joint spectrum of the energy-momentum maps. They
have been recently proved in Ref. \cite{nekonew} from a
geometrical construction. Note also that our starting point here
and in Sec. \ref{mnres} is more general than in Sec. \ref{fhm} in
the sense that only a local structure of the bifurcation diagram
of the energy-momentum map $F=(H,J)$ is assumed.
\end{remark}
As was done for the 1:-2 resonance, we have to determine the
expressions of the functions $\Theta$ and $\tau$ in terms of the
invariant polynomials, the monodromy matrix being given by the
behavior of these two functions in the neighborhood of $C$.
\begin{lemma} \label{lemmatheta}
The functions $\Theta$ and $\tau$ are given for a point $(h,j)\in
\mathcal{R}_{reg}$ by the following integrals
\begin{eqnarray} \label{res1n9}
\left\{ \begin{array}{ll}
\Theta(h,j)=\int_{\pi_1^-}^{\pi_1^+}\frac{h}{j+\pi_1}\frac{d\pi_1}{\sqrt{Q}}+\Theta_0(h,j)\\
\tau(h,j)=\frac{1}{n}\int_{\pi_1^-}^{\pi_1^+}\frac{(j+\pi_1)^{n-1}}{\sqrt{Q}}d\pi_1 \\
\end{array}\right.
\ .
\end{eqnarray}
$\Theta_0$ is a function over $\mathcal{R}$ which is continuous on
the line of singularities $C$ and which therefore gives a trivial
contribution to the monodromy matrix. $\tau$ is the first return
time of the rescaled vector field
$X=\frac{1}{(p_1^2+q_1^2)^{n-1}}{X_H}$.
\end{lemma}
\begin{proof}
By definition, we have for a point $(h,j)\in\mathcal{R}_{reg}$
that \cite{frac4}
\begin{equation} \label{res1n11}
\Theta(h,j)=\int_0^T\dot{\theta}dt \ ,
\end{equation}
where $\theta$ is the angle variable conjugate to $J$. It can be
expressed in terms of the variables $(q_1,p_1)$ by
\begin{equation} \label{res1n12}
\theta=\arg(p_1+iq_1) \ .
\end{equation}
This corresponds to a particular choice of the angle $\theta$.
Other choices lead to the same monodromy matrix. This point will
be detailed in Sec. \ref{mnres} for the $m:-n$ resonance.
Differentiating Eq. (\ref{res1n12}) and using Hamilton's
equations, one arrives to
\begin{equation} \label{res1n13}
\dot{\theta}=\frac{n\pi_3+\frac{\partial R}{\partial
p_1}p_1+\frac{\partial R}{\partial q_1}q_1}{q_1^2+p_1^2} \ ,
\end{equation}
which simplifies into
\begin{equation} \label{res1n14}
\dot{\theta}=\frac{1}{j+\pi_1}(nh-nR+\frac{\partial R}{\partial
p_1}p_1+\frac{\partial R}{\partial q_1}q_1) \ .
\end{equation}
The integral of Eq. (\ref{res1n11}) can be rewritten as an
integral in the reduced phase space $P_j$ \cite{frac4} :
\begin{equation} \label{res1n15}
\Theta(h,j)=2\int_{\pi_1^-}^{\pi_1^+}\dot{\theta}
\frac{d\pi_1}{\dot{\pi_1}} \ .
\end{equation}
Using the particular form of the polynomial $R$ (see lemma
\ref{lemmar}), it can be shown that the last three terms of Eq.
(\ref{res1n14}) give a continuous contribution to the function
$\Theta$ denoted $\Theta_0$. Since $\dot{\pi}_1=2n\pi_2$
\cite{frac3}, we finally obtain that
\begin{equation} \label{res1n16}
\Theta(h,j)=\int_{\pi_1^-}^{\pi_1^+}\frac{h}{j+\pi_1}\frac{d\pi_1}{\sqrt{Q}}+\Theta_0(h,j)
\ .
\end{equation}
The term $\pi_2$ has been replaced in Eq. (\ref{res1n16}) by
combining Eqs. (\ref{res1n3a}) and (\ref{res1n5}).

The determination of $\tau$ is straightforward if we remark that
\begin{equation} \label{res1n17}
\tau(h,j)=\int_{0}^{\tau}ds=\int_0^{T}\frac{ds}{dt}dt=\int_0^T
(p_1^2+q_1^2)^{n-1}dt \ ,
\end{equation}
where $s$ and $t$ are respectively the rescaled and the original time.
The rest of the proof consists, as we did for $\Theta$, in
rewriting the integral of Eq. (\ref{res1n17}) in the reduced phase
space $P_j$ and leads to
\begin{equation} \label{res1n17a}
\tau(h,j)=\frac{1}{n}\int_{\pi_1^-}^{\pi_1^+}\frac{(j+\pi_1)^{n-1}}{\sqrt{Q}}d\pi_1
\ .
\end{equation}\qed
\end{proof}
The last technical point to be discussed is the behavior of the
roots of $Q$ as $h$ go to zero.
\begin{lemma} \label{root1}
The complex discriminant locus $\Delta$ of $Q$ near the origin
$(h=0,j=0)$ is given by
\begin{eqnarray} \label{res1n20}
\Delta=\left\{ \begin{array}{ll}
h=0\\
h=\pm \sqrt{-n^nj^{n+1}\frac{2^{n+1}}{(n+1)^{(n+1)}}} \\
\end{array}\right. \ .
\end{eqnarray}
The polynomial $Q$ as a function of $x$ and in the limit $h\to 0$,
$j<0$ fixed has $n$ roots $x_k$ whose
 leading term reads
\begin{equation} \label{res1n18}
x_k=\frac{h^{2/n}}{(-2j)^{1/n}}e^{2i\pi k/n} \ ,
\end{equation}
with $k\in\{0,1,2,...,n-1\}$. The other roots of $Q$ have a
non-zero finite limit.
\end{lemma}
\begin{proof}
We first determine the complex discriminant locus $\Delta$ of $Q$
near the point $(h=0,j=0)$. Constructing the Newton polyhedron of
$Q$ and taking into account only the terms of lower degrees, the
principal part $Q_N$ of $Q$ can be written
\begin{equation} \label{res1n19}
Q_N(x)=x^{n+1}-2jx^n-h^2 \ .
\end{equation}
A straightforward calculation then leads to $\Delta$.

In the limit $h\to 0$, $j<0$ fixed, we construct the Newton
polygon associated to $Q$. The roots of the principal part of
$Q_N$ satisfy $-2jx^n=h^2$ which allows to deduce the $n$ roots
$x_k$ ($k\in\{0,1,\cdots,n-1\}$).\qed
\end{proof}

We have now all the tools ready to prove proposition \ref{prop1n}.

\begin{proof}
We first consider the case $h>0$ and $j<0$. We recall some of the
properties of the roots of the polynomial $Q$ viewed as a function
of $x$ which will be used in the calculation. We assume that $Q$
has $N$ roots with $N>n$, denoted $x_i$ ($i\in\{0,1,\cdots
N-1\}$). $x_{n},x_{n+1},\cdots,x_{N-1}$ are the roots of $Q$ of
 order 1 in the limit $h\to 0$, $j<0$ fixed. $x_{n}$ is the smallest real
 positive root of $Q$ with a non-zero limit. Since the polynomial $R$ is
defined up to a multiplicative constant, we can write $Q$ without
loss of generality as follows $Q(x)=\prod_{i=0}^{i=N-1}(x-x_i)$.
Examination of the coefficients of $Q$ leads to the following
relations
\begin{eqnarray} \label{res1n20a}
\left\{ \begin{array}{lll}
\prod_{i=0}^{i=N-1}x_i=(-1)^{N+1}h^2\\
\prod_{i=0}^{n-1}x_i=(-1)^{n-1}\frac{h^2}{(-2j)} \\
\prod_{i=n}^{N-1}x_i=2j(-1)^{N-n+1} \\
\end{array}\right. \ .
\end{eqnarray}
We determine only an equivalent of the function $\Theta-\Theta_0$.
We obtain
\begin{equation} \label{res1n21}
\Theta(h,j)-\Theta_0(h,j) \sim
h\int_{\frac{|h|^{2/n}}{(-2j)^{1/n}}}^{x_{n}}\frac{dx}{x\sqrt{\prod_{i=0}^{i=N-1}(x-x_i)}}
\ ,
\end{equation}
where $\frac{|h|^{2/n}}{(-2j)^{1/n}}$ is the smallest positive
real root of $Q$.
 We introduce the function $k$ such that $\frac{|h|^{2/n}}{(-2j)^{1/n}}\ll
k\ll 1$ when $h\to 0$. We then proceed as in the proof for the
resonance 1:-2 by decomposing the integral of Eq. (\ref{res1n21})
into two integrals. Only the first integral from
$\frac{|h|^{2/n}}{(-2j)^{1/n}}$ to $k$ has a limit different from
zero. We then have
\begin{equation} \label{res1n22}
\Theta(h,j)-\Theta_0(h,j) \sim h
\int_{\frac{|h|^{2/n}}{(-2j)^{1/n}}}^k\frac{dx}{x\sqrt{(-2j)(x^n-\frac{h^2}{-2j})\times
1}} \ .
\end{equation}
The change of variables $x=\frac{|h|^{2/n}}{(-2j)^{1/n}}u$ leads to
the following expression for $\Theta-\Theta_0$
\begin{equation} \label{res1n23}
\Theta(h,j)-\Theta_0(h,j) \sim
\int_1^{+\infty}\frac{du}{u\sqrt{u^n-1}} \ .
\end{equation}
Using the fact that
\begin{equation} \label{res1n24}
\int\frac{du}{u\sqrt{u^n-1}}=\frac{2}{n}\arctan [\sqrt{u^n-1}] \ ,
\end{equation}
one finally arrives to
\begin{equation} \label{res1n25}
\lim_{h\to 0^+,j<0} \Theta(h,j)-\Theta_0(h,j)=\frac{\pi}{n} \ .
\end{equation}
It can also be shown that
\begin{equation} \label{res1n26}
\lim_{h\to 0^-,j<0} \Theta(h,j)-\Theta_0(h,j)=\frac{-\pi}{n} \ .
\end{equation}
Following the same arguments, we can calculate the limit of
$\tau$. We decompose $\tau$ into two integrals but we only examine
the last one denoted $\tau_2$. The calculation of the other
integral can be done along the same lines. An equivalent of
$\tau_2$ is given by
\begin{equation} \label{res1n27}
\tau_2(h,j)\sim \frac{1}{n}
\int_{k}^{x_{n}}\frac{x^{n-1}dx}{\sqrt{\prod_{i=0}^{N-1}(x-x_i)}}
\ ,
\end{equation}
which can be rewritten as
\begin{equation} \label{res1n28}
\tau_2(h,j)\sim \frac{1}{n}
\int_{k}^{x_{n}}\frac{dx}{\sqrt{x^{-n+2}(x^{N-n}-2j)}} \ .
\end{equation}
which has the same finite and non-zero limit as $h\to 0^\pm$ with
$j<0$ fixed. \qed
\end{proof}
\subsection{Generalization to $m:-n$ resonances}\label{mnres}
For the $m:-n$ resonant case, the momentum $J$ reads
\begin{equation} \label{resmn1}
J=\frac{1}{2}[m(p_1^2+q_1^2)-n(q_2^2+p_2^2)] \ ,
\end{equation}
where $m\geq 2$ and $n\geq 2$ are relatively prime integers. The
bifurcation diagram corresponding locally near the origin to an
$m:-n$ resonant system is displayed in Fig \ref{fig4}. The
equation of the singular locus is $h=0$. Each point of this line
lifts for $j<0$ to an $n$-curled torus whereas the points of $C$
with $j>0$ lift to an $m$-curled torus.  The origin corresponds to
a pinched-curled torus.
\begin{figure}
\begin{center}
\includegraphics[scale=0.4]{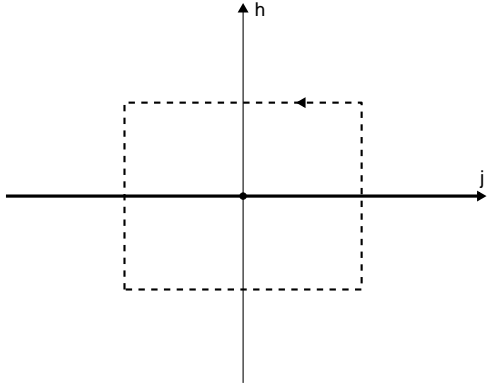}
\caption{\label{fig4} Same as Fig. \ref{fig2} but for the m:-n
resonant system.}
\end{center}
\end{figure}
\begin{proposition} \label{propmn}
The monodromy matrix of a $m:-n$ resonant system of the form
(\ref{res1n4}) is given for $m$ and $n$ relatively prime and for a
counterclockwise loop around the origin (see Fig. \ref{fig3}) by
\begin{eqnarray} \label{mnres2}
M= \left( \begin{array}{cc}
1 & 0 \\
-\frac{1}{mn} & 1
\end{array} \right) \ .
\end{eqnarray}
\end{proposition}
The proof of proposition \ref{propmn} follows the same lines as
for the resonance $1:-n$. As in Sec. \ref{res1n}, we consider the
family of energy-momentum maps given by Eqs. (\ref{res1n4}) where
$\pi_1$, $\pi_3$ and $J$ are given by Eqs. (\ref{res1n3}) with
$m>1$. Some lemmas are required before the final proof.
\begin{lemma} \label{fmn}
The energy-momentum map $F$ of Eq. (\ref{res1n4}) has locally the
bifurcation diagram of Fig. \ref{fig4} in a neighborhood of the
origin if $R$ is of the form
\begin{equation} \label{mnres3}
R(\pi_1,J)=(\pi_1+J)^{n'}(\pi_1-J)^{m'}\tilde{R}(\pi_1,J) \ ,
\end{equation}
where $n'>\frac{n}{2}$, $m'>\frac{m}{2}$. $\tilde{R}$ is a
polynomial such that the two smallest real positive roots of $Q$
which are larger than $j$ are simple roots for
$(h,j)\in\mathcal{R}_{reg}$. $Q$ is the following polynomial
\begin{equation} \label{mnres5}
Q=(\pi_1+j)^n(\pi_1-j)^m-\big(h-R(\pi_1,j)\big)^2 \ .
\end{equation}
The two roots are denoted $\pi_1^-$ and $\pi_1^+$.
\end{lemma}
\begin{proof}
The proof is similar to the proof of lemma \ref{lemmar}.\qed
\end{proof}
\begin{lemma} \label{thetamn}
The functions $\Theta$ and $\tau$ are defined for a point
$(h,j)\in \mathcal{R}_{reg}$ by the following integrals
\begin{eqnarray} \label{mnres4}
\left\{ \begin{array}{ll}
\Theta(h,j)=\int_{\pi_1^-}^{\pi_1^+}[\frac{hu}{j+\pi_1}+\frac{hv}{\pi_1-j}]\frac{d\pi_1}{\sqrt{Q}}+\Theta_0(h,j)\\
\tau(h,j)=\frac{1}{mn}\int_{\pi_1^-}^{\pi_1^+}\frac{(j+\pi_1)^{n-1}(\pi_1-j)^{m-1}}{\sqrt{Q}}d\pi_1 \\
\end{array}\right.
\ .
\end{eqnarray}
$\Theta_0$ is a function over $\mathcal{R}$ which is continuous on
the line of singularities $C$. $u$ and $v$ are two integers such
that $mu-nv=1$. $\tau$ is the first return time of the rescaled
vector field
$X=\frac{1}{(p_1^2+q_1^2)^{n-1}(p_2^2+q_2^2)^{m-1}}{X_H}$ defined
such that $\tau$ has a non-zero finite limit on $C$.
\end{lemma}
\begin{remark}
One can obtain the expression given by Eq. (\ref{res1n9})
corresponding to the case $1:-n$ as a special case of Eq.
(\ref{mnres4}) with $u=1$ and $v=0$, which is a particular
solution of $mu-nv=1$ for $m=1$.
\end{remark}
\begin{proof}
The determination of $\Theta$ is based on the dependence of the
angle $\theta$ as a function of the coordinates
$(p_1,q_1,p_2,q_2)$. To clarify this question, we introduce the
canonical conjugate coordinates $(I_1,\phi_1)$ and $(I_2,\phi_2)$
which are defined as follows
\begin{eqnarray} \label{mnres6}
\left\{ \begin{array}{ll}
q_k=\sqrt{2I_k}\sin\phi_k \\
p_k=\sqrt{2I_k}\cos\phi_k
\end{array}\right.
\ .
\end{eqnarray}
Note that the polar coordinates $(I_k,\phi_k)$ are only defined if
$p_k^2+q_k^2>0$ \cite{arnold} which is not the case on the line
$C$. By definition, the angles $\phi_1$ and $\phi_2$ vary in an
interval of length $2\pi$. We look for a linear canonical
transformation which transforms the two angles $\phi_1$ and
$\phi_2$ into $\theta$ and $\psi$ where the angle $\theta$ is
canonically conjugate to $J$. The angular dependence of $H$ in the
variables $(I_1,\phi_1,I_2,\phi_2)$ is given by the term $\pi_3$
and is equal to $n\phi_1+m\phi_2$. We then set
\begin{eqnarray} \label{mnres7}
\left\{ \begin{array}{ll}
\theta=u\phi_1+v\phi_2 \\
\psi=n\phi_1+m\phi_2
\end{array}\right. \ ,
\end{eqnarray}
where $u$ and $v$ are integers such that $mu-nv=1$ which ensures
that the determinant of the linear canonical transformation is 1
and that $\theta$ and $\psi$ are two angles varying in an interval
of length $2\pi$. Since $m$ and $n$ are relatively prime, the
Bezout theorem states that this equation has a solution
($u_0$,$v_0$). There are an infinite number of solutions which can
be written ($u_0+kn$,$v_0+km$) with $k\in\Z$. We denote by $(u,v)$
one of these solutions. A choice of a couple $(u,v)$ is associated
to a choice of a particular basis of the homology group.

The generating function $F_2$ of type 2 \cite{arnold} associated
to the canonical transformation is given by
\begin{equation} \label{mnres8}
F_2=(u\phi_1+v\phi_2)\tilde{J}+(n\phi_1+m\phi_2)\mathcal{I} \ ,
\end{equation}
where $\tilde{J}$ and $\mathcal{I}$ are the momenta conjugated
respectively to $\theta$ and $\psi$. From the definition of $F_2$,
one deduces that
\begin{eqnarray} \label{mnres9}
\left\{ \begin{array}{ll}
\tilde{J}=mI_1-nI_2 \\
\mathcal{I}=uI_2-vI_1
\end{array}\right. \ ,
\end{eqnarray}
and that as expected $\tilde{J}=J$. We will drop the tilde in the
rest of the proof.

The angle $\theta$ can therefore be written as
\begin{equation} \label{mnres10}
\theta=u\arg(p_1+iq_1)+v\arg(p_2+iq_2) \ .
\end{equation}
Differentiating Eq. (\ref{mnres9}) with respect to time and using
the Hamilton equations, one obtains that
\begin{equation} \label{mnres11}
\dot{\theta}=\frac{u\big(n\pi_3+\frac{\partial R}{\partial
p_1}p_1+\frac{\partial R}{\partial
q_1}q_1\big)}{q_1^2+p_1^2}+\frac{v\big(m\pi_3+\frac{\partial
R}{\partial p_2}p_2+\frac{\partial R}{\partial
q_2}q_2\big)}{q_2^2+p_2^2} \ ,
\end{equation}
which leads to
\begin{equation} \label{mnres11a}
\dot{\theta}=\frac{u\big(nmh-nmR+m\frac{\partial R}{\partial
p_1}p_1+m\frac{\partial R}{\partial
q_1}q_1\big)}{j+\pi_1}+\frac{v\big(nmh-nmR+n\frac{\partial
R}{\partial p_2}p_2+n\frac{\partial R}{\partial
q_2}q_2\big)}{\pi_1-j} \ .
\end{equation}
The last step consists in rewriting this integral as an integral
in the reduced phase space $P_j$. The $R$-dependent part of
$\Theta$ gives a continuous contribution on the line $C$ denoted
$\Theta_0$. Since $\dot{\pi}_1=2mn\pi_2$, one finally obtains
\begin{equation} \label{mnres12}
\Theta(h,j)=\int_{\pi_1^-}^{\pi_1^+}[\frac{hu}{j+\pi_1}+\frac{hv}{\pi_1-j}]\frac{d\pi_1}{\pi_2}+\Theta_0(h,j)
\ .
\end{equation}
For $\tau$ the proof is straightforward and similar to the one of
lemma \ref{lemmatheta}.\qed
\end{proof}
\begin{lemma} \label{lemmarootmn}
The complex discriminant locus $\Delta$ near the origin is given
by
\begin{eqnarray} \label{mnres15}
\Delta=\left\{ \begin{array}{ll}
h=0\\
h=\pm \sqrt{\frac{(-1)^m2^{m+n}m^mn^nj^{m+n}}{(m+n)^{m+n}}} \\
\end{array}\right. \ .
\end{eqnarray}
In the limit $h\to 0$, $j<0$ fixed, the polynomial $Q$ as a
function of $x$ has n roots $x_k$ whose
 leading term is
\begin{equation} \label{mnres13}
x_k=\frac{h^{2/n}}{(-2j)^{m/n}}e^{\frac{2i\pi k}{n}} \ .
\end{equation}
with $k\in\{0,1,\cdots ,n-1\}$, the other roots having a finite
limit different from zero.

In the limit $h\to 0$, $j>0$ fixed, the polynomial $Q$ as a
function of $x'$ has m roots $x'_k$ whose
 leading term is
\begin{equation} \label{mnres13a}
x'_k=\frac{h^{2/m}}{(2j)^{n/m}}e^{\frac{2i\pi k}{m}} \ .
\end{equation}
with $k\in\{0,1,\cdots ,m-1\}$, the other roots having a finite
limit different from zero.
\end{lemma}
\begin{remark}
We notice that the case $j>0$ and $j<0$ give the same expressions
but with $m$ and $n$ interchanged.
\end{remark}
\begin{proof}
Let us assume that $j<0$. As for the $1:-n$ resonance, we
calculate the complex discriminant locus of $Q$ near the origin as
a function of $x$. The construction of the Newton polyhedron gives
the principal part $Q_N$ of $Q$
\begin{equation} \label{mnres14}
Q_N(x)=(x-2j)^mx^n-h^2 \ .
\end{equation}
Simple algebra leads to the discriminant locus $\Delta$.

In the limit $h\to 0$, $j<0$ fixed, the construction of the Newton
polygon of $Q$ leads to the following equation for the roots of
the principal part of $Q$
\begin{equation} \label{mnres16}
(-2j)^mx^n=h^2 \ ,
\end{equation}
and to the $n$ roots $x_k$ of Eq. (\ref{mnres13}). Exchanging the
role of $m$ and $n$ and taking $j>0$, a similar proof gives the
roots of Eq. (\ref{mnres13a}).\qed
\end{proof}

Having established lemmas \ref{fmn}, \ref{thetamn} and
\ref{lemmarootmn} required, we can pass to the proof of
proposition \ref{propmn}.

\begin{proof}
Since the line $C$ is crossed two times by the loop $\Gamma$, the
monodromy matrix $M$ has two contributions denoted $M_-$ for $j<0$
and $M_+$ for $j>0$. We first consider the case $j<0$. The case
$j>0$ will be deduced from the calculation for $j<0$ by exchanging
the role of $m$ and $n$. The calculation is based on the analysis
of the roots of the polynomial $Q$. $Q$ has $N$ roots denoted
$x_i$ with $i\in\{0,1,\cdots N-1\}$. The roots $x_n,\cdots
x_{N-1}$ are of order 1. A simple analysis of the polynomial $Q$
leads to the following relations
\begin{eqnarray} \label{mnres17}
\left\{ \begin{array}{lll}
\prod_{i=0}^{N-1}x_i=h^2(-1)^{N+1}\\
\prod_{i=0}^{n-1}x_i=\frac{h^2}{(-2j)^m}(-1)^{n-1} \\
\prod_{i=n}^{N-1}x_i=(-1)^{N-n+m}(-2j)^m \\
\end{array} \right. \ .
\end{eqnarray}
Following the same steps as in the proof for the resonance $1:-n$,
one arrives to
\begin{equation} \label{mnres18}
\Theta(h,j)-\Theta_0(h,j) \sim uh
\int_{\frac{|h|^{2/n}}{(-2j)^{m/n}}}^k\frac{dx}{x\sqrt{(-2j)^m(x^n-\frac{h^2}{(-2j)^m})\times
1}} \ ,
\end{equation}
where $\frac{|h|^{2/n}}{(-2j)^{m/n}}\ll k \ll 1$. For $h>0$, Eq.
(\ref{mnres18}) simplifies into
\begin{equation} \label{mnres19}
\Theta(h,j)-\Theta_0(h,j) \sim u
\int_{1}^{+\infty}\frac{dx}{x\sqrt{x^n-1}} \ .
\end{equation}
We finally obtain that
\begin{equation} \label{mnres19a}
\lim_{h\to 0^{\pm},j<0}=\Theta(h,j)-\Theta_0(h,j)=\pm
u\frac{\pi}{n} \ .
\end{equation}
A similar proof leads to
\begin{equation} \label{mnres19b}
\lim_{h\to 0^{\pm},j>0}=\Theta(h,j)-\Theta_0(h,j)=\pm
v\frac{\pi}{m} \ .
\end{equation}
The calculation of $\tau$ uses the same arguments and shows that
$\tau$ is continuous on $C$.

One then deduces that the matrices $M_-$ and $M_+$ are
respectively given by
\begin{eqnarray} \label{mnres20}
M_-= \left( \begin{array}{cc}
1 & 0 \\
-\frac{u}{n} & 1
\end{array} \right) \ ,
\end{eqnarray}
and
\begin{eqnarray} \label{mnres21}
M_+= \left( \begin{array}{cc}
1 & 0 \\
\frac{v}{m} & 1
\end{array} \right) \ ,
\end{eqnarray}
where we have taken into account for $M_+$ the fact that the line
$C$ is crossed from $h<0$ to $h>0$. The total monodromy matrix is
given by the product of the matrices $M_-$ and $M_+$
\begin{eqnarray}
M= \left( \begin{array}{cc}
1 & 0 \\
-\frac{u}{n} & 1
\end{array} \right)
\left( \begin{array}{cc}
1 & 0 \\
\frac{v}{m} & 1
\end{array} \right)
=\left( \begin{array}{cc}
1 & 0 \\
-\frac{1}{mn} & 1
\end{array} \right) \ ,
\end{eqnarray}
where the relation $mu-nv=1$ has been used.\qed
\end{proof}
\section{Extension to the complex domain} \label{complex}
\subsection{Idea of the method}
In this section, we reformulate the notion of fractional
hamiltonian monodromy by deforming the loop $\Gamma$ close to the
line $C$, such that it bypasses the line $C$ through the complex
domain. The starting point of the complex approach is given by the
expressions of the functions $\Theta$ and $\tau$ as real
one-dimensional integrals (see for instance lemma \ref{lemma1},
Eqs. (\ref{fhm4}) and (\ref{fhm5}) for the 1:-2 resonance case).
If we consider the complexified variables $h$, $j$ and $\pi_1$
then $\Theta(h,j)$ and $\tau(h,j)$ can be interpreted as complex
 integrals on a line of the complex plane $\pi_1$. Furthermore,
 these two functions can be rewritten as integrals of rational 1-forms
 over a cycle $\delta$  on the Riemann surface defined by
 $\pi_2^2=Q(\pi_1)$. This allows the use of topological
 properties of the Riemann surfaces.

 More precisely, we first introduce a
Riemann surface constructed from the energy-momentum map $F$ and
we determine the Gauss-Manin connection of a complex semi-circle
$\Gamma_C$ bypassing the line $C$ of singularities. As explained
below, we can then construct the extension to the complex domain
of the functions $\Theta$ and $\tau$ along $\Gamma_C$. The
monodromy matrix is determined as in the real approach by the
variation of the functions $\Theta$ and $\tau$ along a loop
$\Gamma$ around the origin. The regularizations of the functions
$\Theta$ and $\tau$ to cross the line $C$ are replaced by their
 continuations along $\Gamma_C$. We obtain the fractional monodromy matrix
 by letting the radius of the complex semi-circle $\Gamma_C$ tend to 0.
\subsection{Extension to the complex domain of the 1:-2
resonant system}\label{complexe12}
\begin{remark}
All that has been established in Sec. \ref{real1m2} for the real
approach can be done exactly in the same way for the complexified
phase space $T^*\C^2$ where $(J,\pi_1,\pi_2,\pi_3)\in \C^4$,
except for the fact that there is no restriction on the values of
$J$ and $\pi_1$. The quotient is taken to be $\C^*\sim S^1\times
\R^*$. The reduction in the complex approach is described in
appendix \ref{appred}. The manifold $\C^4/S^1\times \R^*$ has real
dimension 6.
\end{remark}

We begin by recalling some of the basic elements of the theory of
complex algebraic curves which will be used throughout this
section (see \cite{kirwan} for a comprehensive introduction).

Let $P:\C^2\to \C$ be a polynomial function. The fibers
$P^{-1}(z)$ are complex algebraic curves (hence two-dimensional
real surfaces). There exists a finite set $\Sigma\subset\C$
(essentially critical values of $P$) such that all fibers
$P^{-1}(z)$, $z\in\C\setminus\Sigma$, look alike. Moreover, the
mapping $P:\C^2\setminus(P^{-1}(\Sigma))\to\C\setminus\Sigma$ is a
locally trivial fibration. Hence, given any path
$\gamma:[z_0,z_1]\to\C\setminus \Sigma$, one can identify the
fibers $P^{-1}(\gamma(z_0))$ and $P^{-1}(\gamma(z_1))$. This
identification is not unique but induces a unique identification
between the homology groups of the two fibers. For
$z_0\in\C\setminus\Sigma$, a cycle $\delta(z_0)\in
H_1(P^{-1}(z_0))$ can be transported along any path $\gamma$ in
$\C\setminus\Sigma$ starting at $z_0$ giving thus a family of
cycles $\delta(z)$.  The transport depends only on the homotopy
class of the path $\gamma$ in $\C\setminus\Sigma$. This is the
classical Gauss-Manin connection \cite{AGZV,zoladek}. The
Gauss-Manin monodromy is defined from the Gauss-Manin connection
as an automorphism of the homology group $H_1$ for a loop in
$\C\setminus\Sigma$.

For the 1:-2 resonance (see Sec. \ref{fhm}), the energy-momentum
map of Eqs. (\ref{eq1}) can be treated in the formalism of complex
algebraic curves even if the situation is more complicated than
described above. The difficulty here lies in the fact that the
{\it fibration} is defined only implicitly by
\begin{equation}\label{complex1}
F^\C
:=\pi_2^2-[(\pi_1-j)(\pi_1+j)^2-\big(h-\epsilon(\pi_1^2-j^2)\big)^2]=0
\ ,
\end{equation}
i.e. $\pi_2^2=Q(\pi_1,j,h)$.

For fixed generic values of $(h,j)$, Eq. (\ref{complex1}) defines
a torus from which two points at infinity have been deleted. The
corresponding complex algebraic curve is schematically represented
in Fig. \ref{fig5}. This representation can be understood by
solving Eq. (\ref{complex1}) with respect to
 $\pi_2$. Generically, there are four values of $\pi_1$ for which
$Q(\pi_1,j,h)=0$, giving each a single solution for $\pi_2=0$.
These points are the ramification points of the complex algebraic
curve. For all other points $\pi_1$, there are two solutions
$\pi_2$ of Eq. (\ref{complex1}) represented by the two leaves in
 Fig. \ref{fig5}.
\begin{figure}
\begin{center}
\includegraphics[scale=0.8]{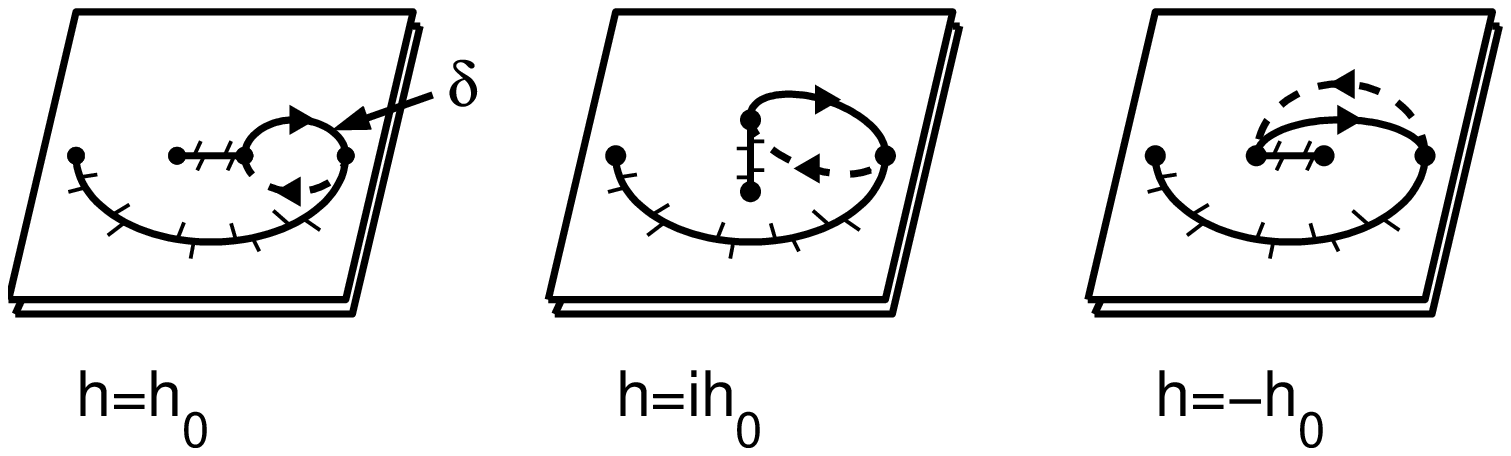}
\includegraphics[scale=0.8]{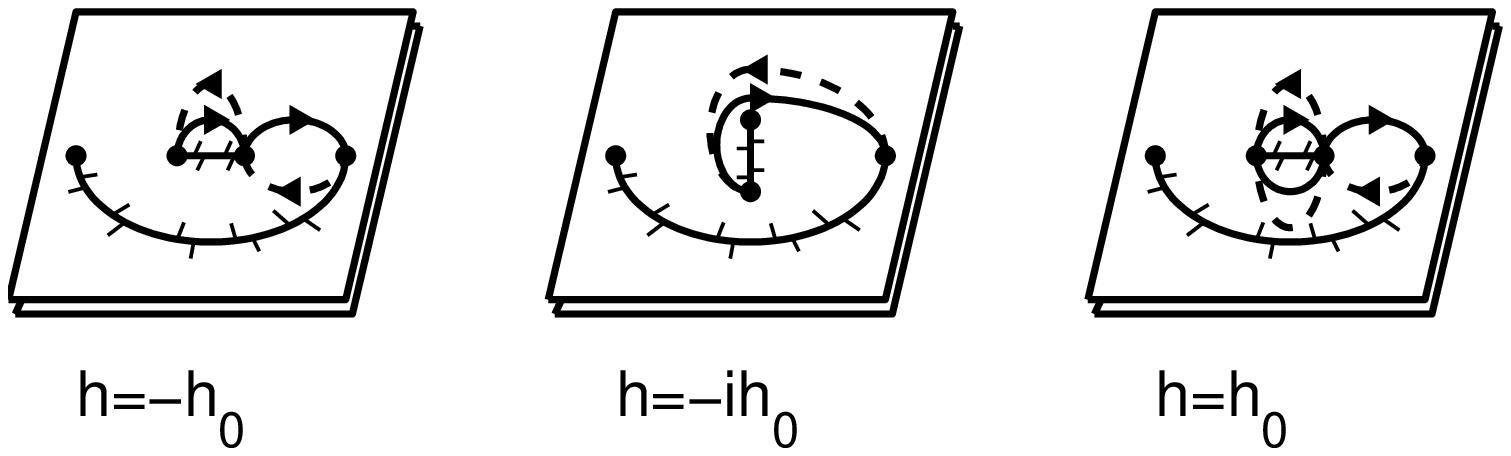}
\caption{\label{fig5} Transport of the cycle $\delta$ along a loop
around the line $C$. The loop is an arc of a circle of radius
$h_0>0$. The large solid lines represent arbitrary branch cuts of
the Riemann surfaces and the full dots ramification points (see
text). The parts in solid and dashed lines of the loop
respectively lie in the upper and lower leaves of the Riemann
surface.}
\end{center}
\end{figure}
We denote by $F^\C _{(j,h)}$ the set of points
$(\pi_1,\pi_2)\in\C^2$ such that $(\pi_1,\pi_2,j,h)$ verify Eq.
(\ref{complex1}). Note that for two different values
$(h,j)\not=(h',j')$ the fibers $F^\C _{(j,h)}$ and $F^\C
_{(j',h')}$ intersect. Nevertheless, it is possible to generalize
the Gauss-Manin connection to this case. Consider for that the
mapping $G:\C^4\to\C^3$, given by
\begin{equation} \label{complex2}
G(\pi_1,\pi_2,h,j)=(\pi_2^2-Q(\pi_1,h,j),h,j) \ .
\end{equation}
The mapping $G$ defines a fibration on the complement
$\C^3\setminus\Sigma$ of the set $\Sigma\in\C ^3$ where its rank
is not maximal. We use here the Ehresmann fibration theorem
\cite{wolf}. We take as basis of our fibration denoted
$\mathcal{B}\in \C^2$ the set $(\{0\}\times
\C^2)\cap(\C^3\setminus\Sigma)$ which is viewed as a set in the
$(h,j)$-space $\C^2$. As in the real approach, we introduce the
variable $x=j+\pi_1$ and we set $y=\pi_2$. The singular locus of
the fibration is the set of points $(h,j)\in \C^2$ where the
polynomial $Q$ has multiple complex roots. This set is given by
Eq. (\ref{fhm8}) and is displayed for $j\in\R$ in Fig. \ref{fig6}.
We notice that for $(h,j)=(0,0)$ the fibration is singular and
that the singularity $(x=0,y=0)$ is not of Morse type.

In the reduced phase space $P_j$, the original real torus projects
to a cycle $\delta(h,j)$ delimited by $\pi_1^-$ and $\pi_1^+$ (see
for instance Fig. \ref{fig3}). $\pi_1^-$ and $\pi_1^+$ are the two
largest real roots of the polynomial $Q$ as a function of $\pi_1$.
Returning back to the Riemann surface and following notations of
Eqs. (\ref{fhm10}), the roots of $Q$, which are simple
ramification points of the Riemann surface, are denoted $x_k$
($k=1,\cdots 4$). One can introduce cuts along the segment
$x_2x_3$ and along a simple curve joining $x_1$ and $x_4$ and
avoiding the segment $x_2x_3$. For $(h,j)\in\mathcal{R}_{reg}$,
the cycle $\delta(h,j)$ is represented by the real oval between
the two largest real ramification points which correspond
respectively to $j+\pi_1^-$ and $j+\pi_1^+$. To be coherent with
the real approach, this cycle is oriented from $x=j+\pi_1^-$ to
$x=j+\pi_1^+$ in the upper leaf and from $x=j+\pi_1^+$ to
$x=j+\pi_1^-$ in the lower one. All these notations are displayed
in Fig. \ref{fig5}.
\subsection{Computation of fractional monodromy from the Gauss-Manin monodromy} \label{complex12def}
We pursue in this section the construction for the 1:-2 resonant
system to arrive to the computation of fractional monodromy at the
end of the section. Let $\Gamma$ be a loop around the origin. We
recall that the computation of the monodromy matrix associated to
$\Gamma$ is based on the difference of the values of the functions
$\Theta$ and $\tau$ at each side of $C$ as $h$ goes to 0. The goal
here is to compute the variations of these functions using their
 extensions to the complex domain near $C$.
\begin{definition}
Starting with the result of lemma \ref{lemma1}, we introduce the
complex continuation of the functions $\Theta$ and $\tau$ defined
by
\begin{eqnarray}
\left\{ \begin{array}{ll}
\Theta(h,j)=\frac{h}{2i\varepsilon}\int_{\delta(h,j)}\frac{dx}{x y} \\
\tau(h,j)=\frac{1}{4i\varepsilon}\int_{\delta(h,j)}\frac{x dx}{y}
\end{array} \right. \ ,
\end{eqnarray}
where $(h,j)\in\mathcal{B}$ and
$y^2=(x-x_1)(x-x_2)(x-x_3)(x-x_4)$. The positive and negative
determinations of the square root $y$ have been chosen
respectively for the upper and the lower leaves of the Riemann
surface. We have added a factor $\frac{1}{2}$ in the definition of
$\Theta$ and $\tau$ to coincide with the real case.
\end{definition}
We locally deform in a neighborhood of the line $C$ the loop
$\Gamma$. We denote by $\Gamma_C$ this complex deformation and by
$\Gamma_R$ the rest of the loop. $\Gamma$, $\Gamma_C$ and
$\Gamma_R$ are represented in Figs. \ref{fig6} and \ref{fig61}.
\begin{remark}
In this work, we have considered deformations of the real loop in
the half-plane $\Im[h]>0$ but they could be equivalently done in
the half-plane $\Im[h]<0$.
\end{remark}
The bypass $\Gamma_C(h_0)$ is a semi-circle of radius $h_0$ in a
plane with $j_0<0$ fixed around the line $C$. The corresponding
real path completing $\Gamma_C(h_0)$ is denoted $\Gamma_R(h_0)$.
From Eq. (\ref{fhm8}) and for a small $j_0$, one deduces that if
$h_0<\sqrt{\frac{-32}{27}j_0^3}$ then $x_2$ and $x_3$ exchange
their positions along $\Gamma_C(h_0)$ whereas if
$h_0>\sqrt{\frac{-32}{27}j_0^3}$ then 3 ramification points of the
Riemann surface ($x_1$, $x_2$ and $x_3$) move and exchange their
positions. The change of the ramification points is displayed in
Fig. \ref{fig7}. It can also be deduced from the asymptotic
expansions of the roots of $Q$ [see Eqs. (\ref{fhm10})]. This can
be seen by parameterizing the loop $\Gamma_C(h_0)$ as
\begin{eqnarray} \label{complex3}
\left\{ \begin{array}{ll}
h=h_0e^{it} \\
j=j_0
\end{array} \right. \ .
\end{eqnarray}
where $t\in [0,\pi]$. There is thus a qualitative difference of
the result depending on the value of $h_0$. $h_0$ must be chosen
sufficiently small to be in the first case since for each fixed
$j_0$ we are interested in the limit $h_0\to 0$.

The family of cycles $\delta(h,j)$ can be obtained by transport of
the cycle $\delta(h_0,j_0)$ along $\Gamma_R$. Examining Fig.
\ref{fig7}, one sees that the lines of singularities crossed by
$\Gamma_R$ have no incidence on the ramification points defining
the cycle $\delta$. The parallel transport of $\delta$ along
$\Gamma_C(h_0)$ is given by the change of the ramification points
along $\Gamma_C(h_0)$. From the Picard-Lefshetz theory
\cite{zoladek,AGZV} (see Figs. \ref{fig5}), one can show that the
cycle $\delta(h_{0},j_{0})$ when transported along the bypass
$\Gamma_C(h_0)$ is transformed into
$\delta(-h_{0},j_{0})+\delta_0(-h_0,j_0)$ where $\delta_0$ is a
vanishing cycle around the ramification points $x_{2}$ and $x_{3}$
of the fiber $F^\C _{(-h_0,j_0)}$. $x_2$ and $x_3$ are defined by
Eqs. (\ref{fhm10a}). For the position of the cuts of Figs.
\ref{fig5}, the cycle $\delta_0(-h_0,j_0)$ is composed of a path
from $x_3$ to $x_2$ on the upper leaf followed by the lift of the
same path to the lower leaf run in the opposite direction. The
vanishing cycles $\delta_0$ and $\delta_1$ are represented in Fig.
\ref{fig62}. The cycle $\delta_1$ will be used in lemma
\ref{lemmares}. Note the different choices of cuts for the Riemann
surface between Figs. \ref{fig5} and Fig. \ref{fig62}.
\begin{figure}
\begin{center}
\includegraphics[scale=0.8]{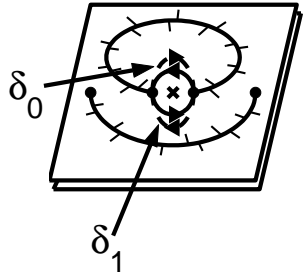}
\caption{\label{fig62} Definition of the cycle $\delta_0$ and
$\delta_1$ for the 1:-2 resonance. The position of the pole of
$\Theta$ is represented by a cross.}
\end{center}
\end{figure}
Hence, an abelian integral
$I(h_0,j_0)=\int_{\delta(h_0,j_0)}\omega$ becomes after going once
around the semi-circle $\Gamma_C(h_0)$ the sum
$I(-h_0,j_0)+\int_{\delta_0(-h_0,j_0)}\omega$ which gives the
variation of the function $I$ over $\Gamma_C(h_0)$. This latter
remark can be applied to the functions $\Theta$ and $\tau$.
However, due to the presence of a pole at $x=0$ for the function
$\Theta$, the cycle $\delta_0$ has to be positioned with respect
to $x=0$. The counterclockwise turning of the ramification points
$x_2$ and $x_3$ implies that the cycle $\delta_0$ avoids the
singularity $x=0$ from above.
\begin{figure}
\begin{center}
\includegraphics[scale=0.4]{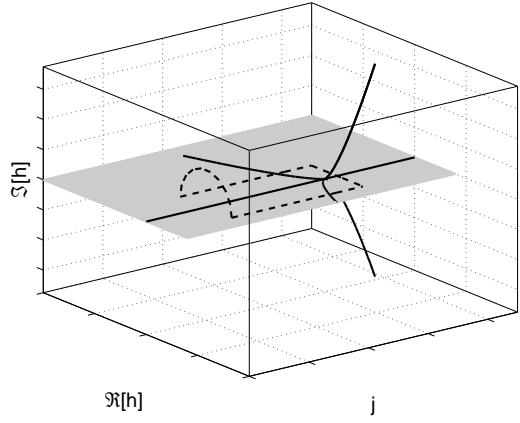}
\caption{\label{fig6} Complex discriminant locus $\Delta$ (solid
lines) of the energy-momentum map of Eq. (\ref{eq1}) for $j\in\R$
and $h\in\C$. The grey plane corresponds to the real bifurcation
diagram. The dashed lines represent the loop $\Gamma$ locally
deformed near $C$ to the complex domain. The arc of circle is in a
complex $h-plane$ with $j$ fixed.}
\end{center}
\end{figure}
\begin{figure}
\begin{center}
\includegraphics[scale=0.4]{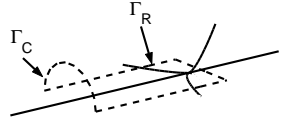}
\caption{\label{fig61} Decomposition of the loop $\Gamma$ into the
paths $\Gamma_C$ and $\Gamma_R$.}
\end{center}
\end{figure}
\begin{figure}
\begin{center}
\includegraphics[scale=0.4]{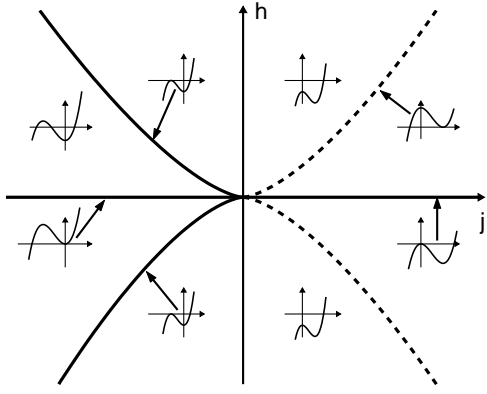}
\caption{\label{fig7} Schematic representation of three of the
roots of the polynomial $Q$ as a function of $h$ and $j$. These
roots are the roots of the principal part $Q_N$ of $Q$ defined by
Eq. (\ref{fhm7}). The polynomial $Q$ has another root larger in
module which is not represented here as it undergoes no
bifurcation. The small inserts depict the graph of $Q_N$ as a
function of $x$ for different values of $h$ and $j$. The position
of the inserts gives the corresponding values of $h$ and $j$. The
solid and dashed lines are lines of singularities of $\Delta$. The
dashed lines do not belong to the real bifurcation diagram (see
Fig. \ref{fig6}).}
\end{center}
\end{figure}

From the analysis of the behavior of the cycle $\delta$ along
$\Gamma$, we can determine the variations of the functions
$\Theta$ and $\tau$ along this loop. These variations are denoted
$\Delta\Theta_{\Gamma}(h_0)$ and $\Delta\tau_{\Gamma}(h_0)$ and
defined as follows
\begin{definition}\label{defvar}
For the energy-momentum map of Eq. (\ref{eq1}),
$\Delta\Theta_{\Gamma}(h_0)$ and $\Delta\tau_{\Gamma}(h_0)$ are
given in the complex approach by
\begin{eqnarray}
\left\{ \begin{array}{ll}
\Delta\Theta_\Gamma(h_0)=\Delta\Theta_{\Gamma_C}(h_0)+\Delta\Theta_{\Gamma_R}(h_0) \\
\Delta\tau_\Gamma(h_0)=\Delta\tau_{\Gamma_C}(h_0)+\Delta\tau_{\Gamma_R}(h_0)
\end{array} \right. \ .
\end{eqnarray}
\end{definition}
A simple calculation allows to simplify the expressions of
$\Delta\Theta_\Gamma(h_0)$ and $\Delta\tau_\Gamma(h_0)$. Since
\begin{equation}
\Delta\Theta_R(h_0)=\frac{h_0}{2i\varepsilon}\int_{\delta(h_0,j_0)}\frac{dx}{x
y}- \frac{-h_0}{2i\varepsilon}\int_{\delta(-h_0,j_0)}\frac{dx}{x
y} \ ,
\end{equation}
and
\begin{equation}
\Delta\Theta_C(h_0)=\frac{-h_0}{2i\varepsilon}\int_{\delta(-h_0,j_0)}\frac{dx}{x
y}+ \frac{-h_0}{2i\varepsilon}\int_{\delta_0(-h_0,j_0)}\frac{dx}{x
y} -\frac{h_0}{2i\varepsilon}\int_{\delta(h_0,j_0)}\frac{dx}{x y}
\ ,
\end{equation}
one deduces for $j_0<0$ fixed that
\begin{eqnarray}\label{res100}
\left\{ \begin{array}{ll}
\Delta\Theta_\Gamma(h_0)=-\frac{h_0}{2i\varepsilon}\int_{\delta_0(-h_0,j_0)}\frac{dx}{x y} \\
\Delta\tau_\Gamma(h_0)=\frac{1}{4i\varepsilon}\int_{\delta_0(-h_0,j_0)}\frac{x
dx}{y}
\end{array} \right. \ .
\end{eqnarray}
We introduce for the function $\Theta$ the following quantities
\begin{eqnarray}\label{resnew1}
\left\{ \begin{array}{lll}
\Delta\Theta_\Gamma=\lim_{h_0\to 0}\Delta\Theta_\Gamma(h_0) \\
\Delta\Theta_{\Gamma_C}=\lim_{h_0\to
0}\Delta\Theta_{\Gamma_C}(h_0) \\
\Delta\Theta_{\Gamma_R}=\lim_{h_0\to
0}\Delta\Theta_{\Gamma_R}(h_0)
\end{array} \right. \ ,
\end{eqnarray}
and the same for $\tau$.

Different results have to be established before computing the
monodromy matrix. The variation of $\Theta$ around $C$ can be
viewed as a residue.
\begin{lemma}\label{lemmares}
The sum of the integrals of the 1-form
$\frac{h_0}{2i\varepsilon}\frac{dx}{xy}$ over $\delta_0$ and
$\delta_1$ is independent of $h_0$ and $j_0$ and equal to
\begin{equation} \label{eqnew1}
\frac{h_0}{2i\varepsilon}\int_{\delta_0(h_0,j_0)}\frac{dx}{x
y}+\frac{h_0}{2i\varepsilon}\int_{\delta_1(h_0,j_0)}\frac{dx}{x
y}=2\pi \ .
\end{equation}
In this equation, $h_0$ is taken sufficiently small and $j_0<0$.
\end{lemma}
\begin{proof}
We use the notations of Fig. \ref{fig62}. The union of $\delta_0$
and $\delta_1$ corresponds to two loops around the pole $x=0$
lying respectively in the upper and the lower leaves of the
Riemann surface. The orientation of these two loops is on the
lower leaf the opposite to the one on the upper leaf. The same
applies to the determination of the square root $y$ for the two
leaves of the Riemann surface. Hence, the sum of the left
hand-side of Eq. (\ref{eqnew1}) is given by two times the residue
of the 1-form $\frac{dx}{xy}$ at $x=0$. One deduces that
\begin{equation} \label{eqnew2}
\frac{h_0}{2i\varepsilon}\int_{\delta_0(h_0,j_0)}\frac{dx}{x
y}+\frac{h_0}{2i\varepsilon}\int_{\delta_1(h_0,j_0)}\frac{dx}{x
y}=\frac{h_0}{2i\varepsilon}4\pi i\textrm{Res}(\frac{1}{xy},x=0) \
.
\end{equation}
Simple algebra leads to
\begin{equation}
\textrm{Res}(\frac{1}{xy},x=0)=\frac{\varepsilon}{h_0} \ ,
\end{equation}
which completes the proof.\qed
\end{proof}
\begin{lemma} \label{lemmacomp12}
The variations of the functions $\Theta$ and $\tau$ along $\Gamma$
are
\begin{eqnarray}
\left\{ \begin{array}{ll}
\Delta\Theta_\Gamma=\pi \\
\Delta \tau_\Gamma=0
\end{array} \right. \ .
\end{eqnarray}
\end{lemma}
\begin{proof}
We use Eqs. (\ref{res100}) and we calculate these two quantities
in the limit $h_0\to 0$ and $j_0<0$ fixed. The asymptotic
expansion of the roots $x_i$ for $h\to 0$ and $j<0$ fixed is given
by Eqs. (\ref{fhm10a}) [see Ref. \cite{frac4} for the explicit
computation].

We begin by the computation of $\Delta\Theta_\Gamma$. We first
notice that $x_2$ and $x_3$ are real. $\delta_0$ can thus be
viewed as a real loop which is locally deformed in a neighborhood
of $x=0$ to avoid the pole in $x=0$. Since the path $\delta_0$ is
oriented in the opposite direction and the sign of $y$ is the
opposite in the lower leaf with respect to the upper leaf, it is
straightforward to see that the contributions of the upper and
lower leaves of the Riemann surface coincide.
$\Delta\Theta_\Gamma$ can be written as follows
\begin{equation}
\Delta\Theta_\Gamma=\lim_{h_0\to 0}\big[ PV
\frac{-h_0}{i\varepsilon}\int_{x_3}^{x_2}\frac{dx}{xy}
+\frac{h_0}{i\varepsilon}\frac{1}{2}2\pi
i\textrm{Res}(\frac{1}{xy},x=0)\big] \ ,
\end{equation}
where $PV$ denotes the principal value of the integral. The
introduction of the principal value is due to the presence of the
pole at $x=0$. The residue is calculated with the positive
determination of the square root $y$. The factor $\frac{-1}{2}$ in
front of the residue corresponds to the fact that the integral is
taken on a semi-circle which is oriented in a clockwise manner.
The contribution of the residue term to $\Delta\Theta_\Gamma$ is
equal to $\pi$. Following Ref. \cite{frac4}, the computation of
the principal value term can be done by using elliptic integrals.
It can be shown that this term is zero.

Using the same arguments, we can deduce that $\Delta\tau_\Gamma=0$
since $\tau$ has no singularity along the
 real segment $x_2x_3$.\qed
\end{proof}
In the limit $h_0\to 0$, $\delta_0$ and $\delta_1$ play a
symmetrical role for $\Theta$. More precisely, we have
\begin{corollary}
The integrals of the 1-form
$\frac{h_0}{2i\varepsilon}\frac{dx}{xy}$ over $\delta_0$ and
$\delta_1$ are given by
\begin{equation}
\lim_{h_0\to
0}\frac{h_0}{2i\varepsilon}\int_{\delta_0(h_0,j_0)}\frac{dx}{x
y}=\lim_{h_0\to
0}\frac{h_0}{2i\varepsilon}\int_{\delta_1(h_0,j_0)}\frac{dx}{x
y}=\pi \ .
\end{equation}
\end{corollary}
\begin{proof}
The proof of lemma \ref{lemmacomp12} has already shown that
\begin{equation}
\lim_{h_0\to
0}\frac{h_0}{2i\varepsilon}\int_{\delta_0(h_0,j_0)}\frac{dx}{x
y}=\pi \ .
\end{equation}
We conclude for $\delta_1$ by using lemma \ref{lemmares}.\qed
\end{proof}
\begin{remark}
We remark that the complex continuations of the functions $\Theta$
and $\tau$ along $\Gamma_C$ have replaced the regularizations of
these functions in the real approach. The semi-circle $\Gamma_C$
is taken to be asymptotic in order for $\Delta\Theta_\Gamma$ and
$\Delta\tau_\Gamma$ to be independent of $h$ and $j$ and to
recover the topological character of fractional monodromy. In
contrast, if we consider a loop around the line $C$ then the
variation of $\Theta$ along this loop is topological as it is
calculated from a residue (see lemma \ref{lemmares}).
\end{remark}
From lemma \ref{lemmacomp12}, we can finally conclude by the
following proposition.
\begin{proposition}
The monodromy matrix $M$ for the loop $\Gamma$ is given by
\begin{eqnarray}
M= \left( \begin{array}{cc}
1 & 0 \\
-\frac{1}{2} & 1
\end{array} \right) \ ,
\end{eqnarray}
\end{proposition}
\begin{proof}
We use lemma \ref{lemmacomp12} and the fact that the monodromy
matrix is determined by the variations $\Delta\Theta_\Gamma$ and
$\Delta\tau_\Gamma$.\qed
\end{proof}
\begin{remark}
In the computation of lemma \ref{lemmacomp12}, we have taken
arbitrary $j$, which shows the topological character of the
definition of fractional monodromy.
\end{remark}
We note that the real and the complex approach can be related by
the following corollary.
\begin{corollary} \label{coro}
$\Delta\Theta_C=0$ and hence $\Delta\Theta_R=\Delta\Theta_\Gamma$
\end{corollary}
\begin{proof}
We apply the Jordan's lemma to show that $\Delta\Theta_C=0$.\qed
\end{proof}
The monodromy matrix in the real approach (resp. complex approach)
is given by $\Delta\Theta_R$ and $\Delta\tau_R$ (resp.
$\Delta\Theta_\Gamma$ and $\Delta\tau_\Gamma$). The corollary
\ref{coro} shows that these two approaches are equivalent.
\subsection{Generalization to $1:-n$ resonance}\label{compl1n}
All the arguments used for constructing the extension to the
complex domain of 1:-2 resonant systems can be generalized to
$1:-n$ and $m:-n$ resonant systems. As for the real approach, we
consider the family $\mathcal{F}$ of energy-momentum maps
introduced in Eqs. (\ref{res1n4}).

The Riemann surface is defined from the relation deduced from Eqs.
(\ref{res1n3a}) and (\ref{res1n4})
\begin{equation}
y^2=x^n(x-2j)-[h-R(x-j,j)]^2 \ .
\end{equation}
The discriminant locus $\Delta$ of these systems is given locally
by Eq. (\ref{res1n20}). Note that this locus is qualitatively
different according to the parity of $n$. More precisely, the real
lines of singularities for $n$ odd become purely imaginary for $n$
even. This does not change the discussion of this section.
Following the preceding case, we consider a loop $\Gamma$ around
the origin which decomposes into a complex semi-circle $\Gamma_C$
around the line $C$ and a real path $\Gamma_R$. Along $\Gamma_C$,
one sees by using the expansion of the roots of $Q$ (lemma
\ref{root1}) that for $j<0$, $n$ roots in the variable $x$ turn
asymptotically around the origin by an angle $\frac{2\pi}{n}$. The
other roots stay fixed to first order in $h$. Fig. \ref{fig88}
illustrates this point for the resonance 1:-3. Note that only the
principal part of the roots are exactly exchanged among each
other.
\begin{figure}
\begin{center}
\includegraphics[scale=0.4]{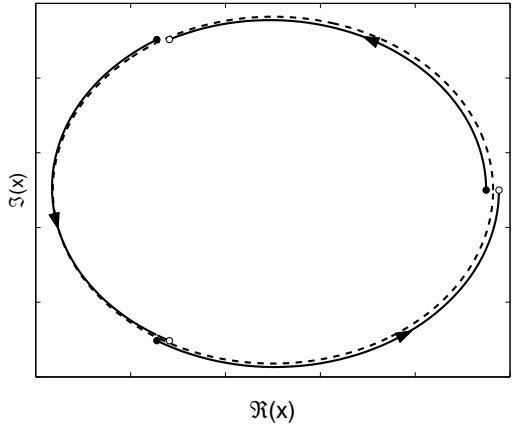}
\caption{\label{fig88} Evolution of three of the roots (see text)
of the polynomial $Q$ from the resonance 1:-3 and for a complex
semi-circle going for $h_0>0$ to $-h_0<0$ ($j_0<0$ fixed).
Numerical values are taken to be $h_0=0.005$ and $j_0=-1$. The
radius of the semi-circle is $h_0$. The energy-momentum map is
given by Eqs. (\ref{res1n5}). The full and open dots represent
respectively the roots for the starting and the ending points of
the path. The dashed line is a circle of radius
$\frac{h_0^{2/3}}{(-2j)^{1/3}}$ which corresponds to the leading
term of the expansion of the roots as $h_0\to 0$.}
\end{center}
\end{figure}
\begin{figure}
\begin{center}
\includegraphics[scale=0.8]{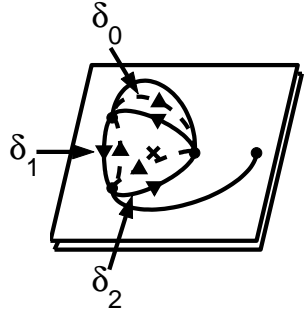}
\caption{\label{fig77} Definition of the cycles $\delta_0$,
$\delta_1$ and $\delta_2$ for the 1:-3 resonance. The cross
indicates the position of the pole of $\Theta$.}
\end{center}
\end{figure}
\begin{figure}
\begin{center}
\includegraphics[scale=0.6]{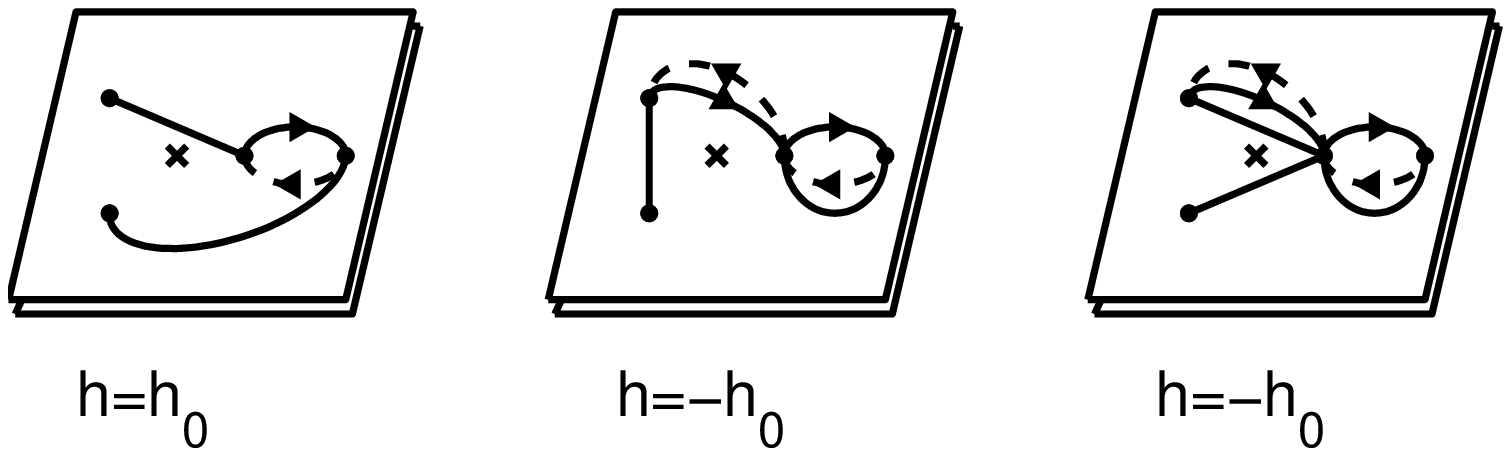}
\includegraphics[scale=0.6]{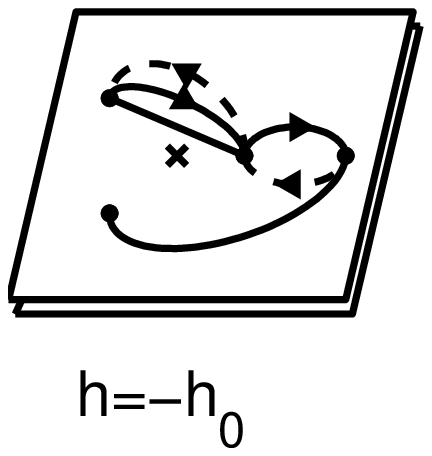}
\caption{\label{fig78} Transport of the cycle $\delta$ along a
semi-circle around the line $C$. The radius of the semi-circle is
$h_0$. The last three figures are equivalent but with different
cuts. The cross indicates the position of the pole of $\Theta$. }
\end{center}
\end{figure}
\begin{figure}
\begin{center}
\includegraphics[scale=0.6]{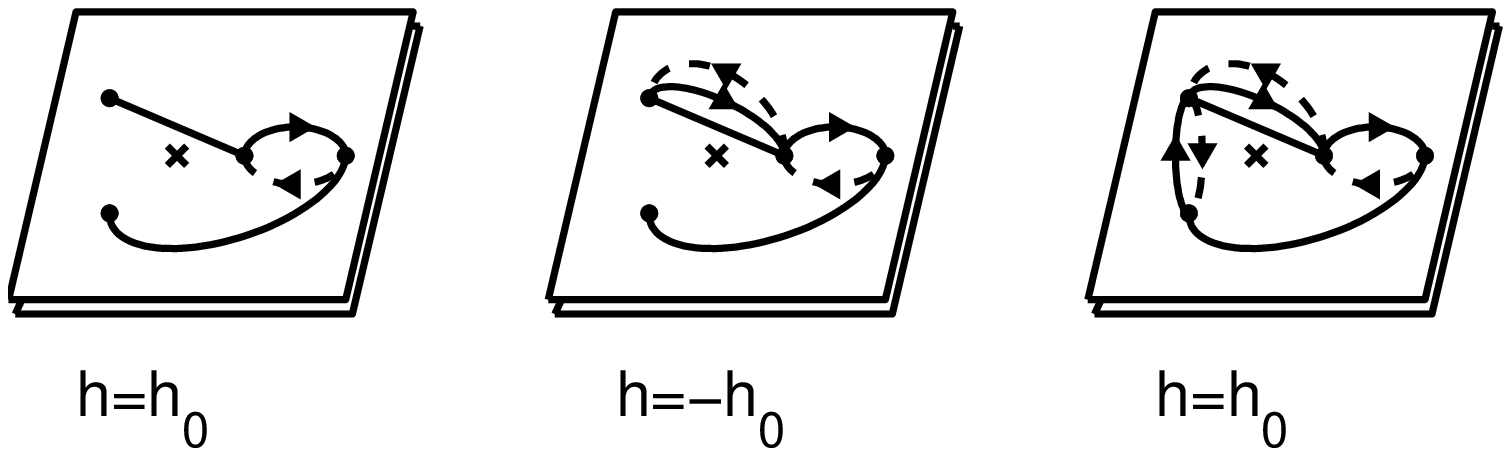}
\includegraphics[scale=0.6]{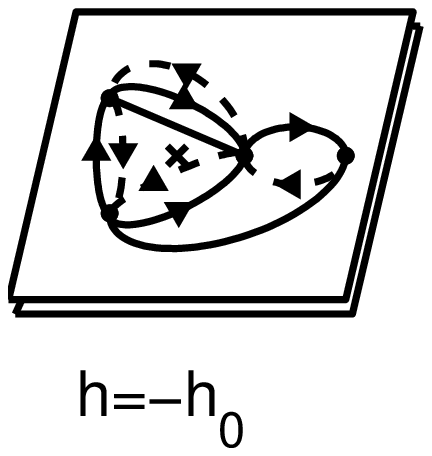}
\caption{\label{fig80} Same as Fig. \ref{fig78} but for 3
semi-circles around the line $C$.}
\end{center}
\end{figure}
Following the change of the ramification points, we can transport
the cycle $\delta$ along $\Gamma_C$. $\delta(h,j)$ is a real oval
between $x_0$ and the ramification point of the Riemann surface
associated to $\pi_1^+$ denoted $x_n$. This cycle is oriented from
$x_0$ to $x_n$ in the upper leaf and from $x_n$ to $x_0$ in the
lower one. Figures \ref{fig78} and \ref{fig80} display the
transport of $\delta$. Not all the ramification points are
represented in Figs. \ref{fig78} and \ref{fig80}. The position of
the cuts is arbitrary but indicates on which leaf of the Riemann
surface a path lies. We can compare two Riemann surfaces if the
cuts of the two surfaces are the same. Since for the 1:-n
resonance the cuts move after a semi-circle, one has to modify the
cuts to recover the initial choice of cuts. An example of this
deformation is given in Figs. \ref{fig78}.

 We see that after the semi-circle $\Gamma_C(h_0)$ around $C$, $\delta(h_0,j_0)$ is
transformed into $\delta(-h_0,j_0)+\delta_0(-h_0,j_0)$. $\delta_0$
is composed of a path from $x_0$ to $x_{n-1}$ in the upper leaf
and of the lift of the same path in the lower leaf but run in the
opposite direction.
\begin{definition}
We define the complex extension of the functions $\Theta$ and
$\tau$ as
\begin{equation} \label{eqcomplex1n00}
\Theta(h,j)=\frac{h}{2}\int_{\delta(h,j)}\frac{dx}{xy}+\Theta_0(h,j) \ ,
\end{equation}
where the positive determination of the square root $y$ is
associated to the upper leaf and
\begin{equation} \label{eqcomplex1n5}
\tau(h,j)=\frac{1}{2n}\int_{\delta(h,j)}\frac{x^{n-1}dx}{y} \ .
\end{equation}
$\Theta_0$ is the complex extension of the real function
$\Theta_0$ introduced in Eq. (\ref{res1n16}). $\Theta_0$ has a
trivial contribution to the monodromy matrix.
\end{definition}
We define the variations $\Delta\Theta_\Gamma(h_0)$ and
$\Delta\tau_\Gamma(h_0)$ of the functions $\Theta$ and $\tau$
along $\Gamma$ as in the definition \ref{defvar} for the 1:-2
resonance. These variations are given by integrals over the cycle
$\delta_0$. $\Delta\Theta_\Gamma$ and $\Delta\tau_\Gamma$ denote
the limits of these variations as $h_0\to 0$. They are given by
\begin{eqnarray}
\left\{ \begin{array}{ll}
\Delta\Theta_\Gamma=\lim_{h_0\to 0} \Delta\Theta_\Gamma(h_0) \\
\Delta\tau_\Gamma=\lim_{h_0\to 0} \Delta\tau_\Gamma(h_0)
\end{array} \right. \ .
\end{eqnarray}
Equivalent lemmas to lemmas \ref{lemmares} and \ref{lemmacomp12}
for the $1:-2$ resonance can be established for the $1:-n$
resonance. We only state the final result.
\begin{proposition} \label{complln}
The monodromy matrix associated to the loop
$\Gamma$ is equal to
\begin{eqnarray}
M= \left( \begin{array}{cc}
1 & 0 \\
-\frac{1}{n} & 1
\end{array} \right) \ .
\end{eqnarray}
\end{proposition}
\begin{proof}
We compute the different integrals in the limit $h_0\to 0$ with
$j_0<0$ fixed. $\Delta\Theta_{\Gamma}(h_0)$ can be written as
\begin{equation} \label{eqcomplex1n0}
\Delta\Theta_{\Gamma}(h_0)=\frac{-h_0}{2}\int_{\delta_0(-h_0,j_0)}\frac{dx}{xy}
\ ,
\end{equation}
where $\delta_0$ is the cycle between $x_0$ and $x_{n-1}$ which is
oriented from $x_0$ to $x_{n-1}$ in the upper leaf and inversely
in the lower leaf. In the limit $h_0\to 0$, one deduces that
\begin{equation} \label{eqcomplex1n1}
\Delta\Theta_{\Gamma}=\lim_{h_0\to 0}
h_0\int^{x_0}_{x_{n-1}}\frac{dx}{x\sqrt{\prod_{k=0}^{N-1}(x-x_k)}}
\ ,
\end{equation}
where the integral is taken along an arc of a circle from
$x_{n-1}$ to $x_0$ of radius $\frac{h_0^{2/n}}{(-2j)^{1/n}}$.
Using the change of variables
$x=\frac{h_0^{2/n}}{(-2j)^{1/n}}e^{i\chi}$, we obtain
\begin{equation} \label{eqcomplex1n2}
\Delta\Theta_{\Gamma}=
\int_{0}^{2\pi/n}\frac{d\chi}{\sqrt{\prod_{k=0}^{n-1}(e^{i\chi}-e^{2i\pi
k/n})}} \ .
\end{equation}
It is straightforward to check that the integrand is not modified
by the translation $\chi'=\chi+\frac{2\pi}{n}$. One then deduces
that
\begin{equation} \label{eqcomplex1n3}
\Delta\Theta_{\Gamma}=\lim_{h_0\to 0} -\frac{1}{n}\textrm{Res}[
\frac{-h_0}{x\sqrt{\prod_{k=0}^{N-1}(x-x_k)}},x=0] \ .
\end{equation}
Simple algebra finally gives
\begin{equation} \label{eqcomplex1n4}
\Delta\Theta_{\Gamma}=\frac{2\pi}{n} \ .
\end{equation}
Similar arguments show that $\Delta\tau_{\Gamma}=0$. We finally
construct the monodromy matrix from the variations of $\Theta$ and
$\tau$ along $\Gamma$. \qed
\end{proof}

We finish this section by presenting a complementary computation
of fractional hamiltonian monodromy in the complex approach. The
idea is here to use direct asymptotic computations to determine
$\Delta\Theta_\Gamma$.
\begin{lemma}
For $j<0$ fixed and $|h|\to 0$, we have the following asymptotic
behavior
\begin{equation}
\Theta(h,j)-\Theta_0(h,j)\sim \frac{2}{n}\arctan{[\frac{k_0}{h}]}
\ ,
\end{equation}
where $k_0$ is a function of $h$ such that $\lim_{|h|\to
0}\frac{|k_0|}{|h|}=0$.
\end{lemma}
\begin{proof}
We proceed as in the real approach by using the asymptotic
expansions of the roots $x_k$ of the polynomial $Q$ (see lemma
\ref{root1}). Following computations of the proof of proposition
\ref{prop1n}, we obtain
\begin{equation}
\Theta(h,j)-\Theta_0(h,j) \sim
h\int_{\frac{h^{2/n}}{(-2j)^{1/n}}}^{x_{n}}\frac{dx}{x\sqrt{\prod_{i=0}^{i=N-1}(x-x_i)}}
\ ,
\end{equation}
which transforms into
\begin{equation}
\Theta(h,j)-\Theta_0(h,j) \sim
h\int_{\frac{h^{2/n}}{(-2j)^{1/n}}}^k\frac{dx}{x\sqrt{(-2j)(x^n-\frac{h^2}{-2j})}}
\ .
\end{equation}
The real function $k$ fulfills $\lim_{|h|\to
0}\frac{k(-2j)^{1/n}}{|h|^{2/n}}=0$. Introducing the variable $u$
such that $x=\frac{h^{2/n}}{(-2j)^{1/n}}u$, one arrives to
\begin{equation}
\Theta(h,j)-\Theta_0(h,j) \sim
\int_1^{\frac{k(-2j)^{1/n}}{h^{2/n}}}\frac{du}{u\sqrt{u^n-1}} \ ,
\end{equation}
where we have used the fact that $\sqrt{h^2}=h$. One finally
obtains that
\begin{equation}
\Theta(h,j)-\Theta_0(h,j) \sim
\big[\frac{2}{n}\arctan{[\sqrt{u^n-1}]}\big]_1^{\frac{k(-2j)^{1/n}}{h^{2/n}}}
\ ,
\end{equation}
which leads to
\begin{equation}
\Theta(h,j)-\Theta_0(h,j) \sim \frac{2}{n}\arctan{[\frac{k_0}{h}]}
\ ,
\end{equation}
where $k_0=k^{n/2}\sqrt{-2j}$.\qed
\end{proof}
The behavior of the function $\Theta$ near the line $C$ is thus
related to the complex function $\arctan$. Using the fact that
\begin{equation}
\arctan z=\frac{1}{2i}\ln[\frac{1+iz}{1-iz}] \ ,
\end{equation}
for $z\in \C$, we can construct the Riemann surface of this
function. This surface has two leaves and a cut between the points
$z=i$ and $z=-i$. Figure \ref{fig89} displays this surface in the
variable $z$. We recall that this function has a jump of $\pi$
along a loop crossing once the cut.
\begin{figure}
\begin{center}
\includegraphics[scale=0.6]{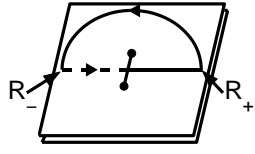}
\caption{\label{fig89} Riemann surface of the function $\arctan$.
The solid and dashed lines respectively lie in the upper and the
lower leaves. The ramification points correspond to the complex
numbers i and -i.}
\end{center}
\end{figure}
Using Eqs. (\ref{res1n23}) and (\ref{res1n24}), one sees that the
limits of $\Theta$ when $h\to 0^\pm$ correspond respectively to
the points $R_+$ and $R_-$ of the Riemann surface (see Fig.
\ref{fig89}). Since these two limit points are the initial and
final points of $\Gamma_R$, we can associate to $\Gamma_R$ the
real path from $R_-$ to $R_+$ of the Riemann surface of the
function $\arctan$. The complex continuation of the function
$\Theta$ along the small semi-circle $\Gamma_C$ of the bifurcation
diagram is associated to the big circle at infinity of Fig.
\ref{fig89}. From Fig. \ref{fig89}, we thus recover that after a
real loop $\Gamma$ locally deformed to the complex domain to
bypass the line $C$ the function $\Theta$ has a jump of
$\frac{2\pi}{n}$. This result is also coherent with the analysis
performed with the variation of the cycles $\delta$. We denote by
$\delta_k$ the cycle between the ramification points $x_k$ and
$x_{k+1}$. The cycle is oriented from $x_k$ to $x_{k+1}$ in the
upper leaf and from $x_{k+1}$ to $x_k$ in the lower one. Following
the proof of proposition \ref{complln}, it can be shown that if we
compute $\Theta$ along one of these cycles then we obtain
asymptotically the same result. More precisely, we have
\begin{equation}
\frac{h}{2}\int_{\delta_k(h,j)}\frac{dx}{xy}=\frac{2\pi}{n} \ ,
\end{equation}
for $j<0$ fixed and $h\to 0^-$. We consider now $r$ semi-circles
around $C$. The cycle $\delta$ is transformed into
$\delta+\delta_0-\delta_1+\delta_2-\cdots+(-1)^{r-1}\delta_{r-1}$.
This is displayed for $r=3$ in Figs. \ref{fig80}. As expected, we
thus see that after an even (resp. odd) number of semi-circles,
the function $\Theta$ has no jump (resp. a jump of
$\frac{2\pi}{n}$) which corresponds to the behavior of the complex
$\arctan$.
\subsection{Generalization to $m:-n$ resonance}\label{complmn}
We now study the $m:-n$ resonant system with $m>1$ and $m$ and $n$
relatively prime. We associate to such a system a Riemann surface
which can be constructed in the variables $x=\pi_1+j$ or
$x'=\pi_1-j$ from
\begin{equation}\label{compmn1}
y^2=x^n(x-2j)^m-[h-R(x-j,j)]^2 \ ,
\end{equation}
or
\begin{equation}\label{compmn2}
y^2=(x'+2j)^nx'^m-[h-R(x'+j,j)]^2 \ .
\end{equation}
One passes from one representation to the other by the relation
$x=x'-2j$. In particular, the two surfaces have the same
ramification points translated by $2j$. Depending on the line of
singularities considered, one or the other surface will be used,
i.e., the surface in $x$ for $j<0$ and the surface in $x'$ for
$j>0$. We next recall that the discriminant locus $\Delta$ of such
a system is given by Eqs. (\ref{mnres15}). From the expansion of
the roots of the polynomial $Q$ in $h=0$ (lemma
\ref{lemmarootmn}), we deduce that for $j<0$ (resp. $j>0$), $n$
(resp. $m$) roots in $x$ (resp. in $x'$) exchange their positions
along a loop around the line $C$. We locally deform $\Gamma$ in a
neighborhood of the line $C$ and we decompose this loop into four
loops $\Gamma_{R_1}$, $\Gamma_{R_2}$, $\Gamma_{C_1}$ and
$\Gamma_{C_2}$ where $\Gamma_{C_1}$ and $\Gamma_{C_2}$ are
respectively two semi-circles around the lines of singularities
$(h=0,j<0)$ and $(h=0,j>0)$. $\Gamma_{R_1}$ and $\Gamma_{R_2}$
complete the loop $\Gamma$ respectively for $h<0$ and $h>0$. We
define the variations of the functions $\Theta$ and $\tau$ along
$\Gamma$ as follows
\begin{eqnarray}
\left\{ \begin{array}{ll}
\Delta\Theta_\Gamma=\Delta\Theta_{\Gamma_{C_1}}+\Delta\Theta_{\Gamma_{C_2}}+\Delta\Theta_{\Gamma_{R_1}}+\Delta\Theta_{\Gamma_{R_2}} \\
\Delta\tau_\Gamma=\Delta\tau_{\Gamma_{C_1}}+\Delta\tau_{\Gamma_{C_2}}+\Delta\tau_{\Gamma_{R_1}}+\Delta\tau_{\Gamma_{R_2}} \\
\end{array} \right. \ ,
\end{eqnarray}
where the different variations are determined asymptotically,
i.e., when the radii of the semi-circles $C_1$ and $C_2$ go to 0.
A simple calculation then show that the study of the $m:-n$
resonance can
 be reduced to the study of two cases for $j>0$ and $j<0$
similar to the $1:-n$ resonance. This allows to compute the jump
of the function $\Theta$ along $\Gamma$ and to deduce the
monodromy matrix for the real loop $\Gamma$. We recover the result
of proposition \ref{propmn}.
\section{Conclusion} \label{conc}
In this paper, we have investigated the notion of Fractional
Hamiltonian Monodromy in the 1:-2, $1:-n$ and $m:-n$ resonant
systems. We have discussed our asymptotic method to calculate the
monodromy matrix in the real approach and we have proposed a
definition of fractional monodromy using a complex extension of
the bifurcation diagram. From this definition, we have recovered
the results of the real approach. At this point, the question
which naturally arises is the generalization of this concept to
other types of singularities. A part of the answer could be given
by applying the complex approach to a new generalization of
standard monodromy, the bidromy which has been recently introduced
in Ref. \cite{bidromy}.
\appendix \label{app}
\section{Geometric construction}\label{appgeom}
\subsection{1:-2 resonance} \label{app12}
As stated in the introduction, the geometric construction of
fractional monodromy generalizes the construction of standard
monodromy. We propose a schematic representation that is slightly
different from the original one proposed in Ref. \cite{frac2}.
This representation has the advantage to be easily generalizable
to 1:-n resonant systems. We use a standard representation of a
torus i.e. a rectangle whose edges are identified according to the
arrows. The crossing of the line $C$ is displayed in Figs.
\ref{figapp1} and \ref{figapp2}. A curled torus is represented by
two rectangles glued along a horizontal edge and with a particular
identification of the vertical edges (see Figs. \ref{figapp1},
$h=0$). This construction is not limited to a given
energy-momentum map but can be applied to any energy-momentum map
having a bifurcation diagram with such a line of singularities.
\begin{figure}
\begin{center}
\includegraphics[scale=0.6]{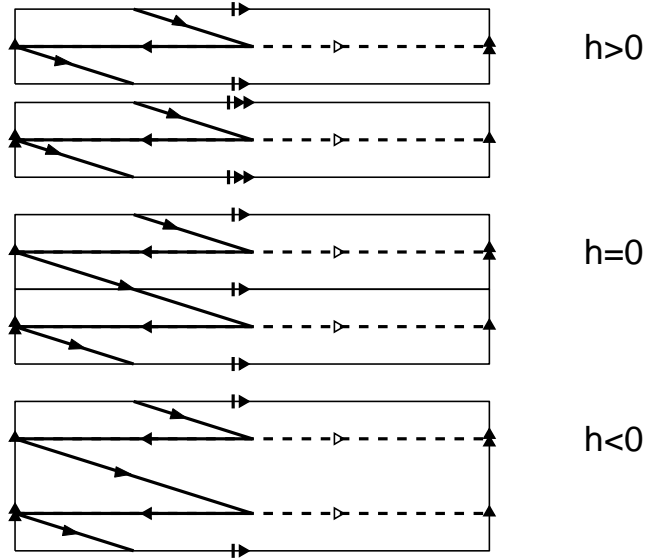}
\caption{\label{figapp1} Schematic representation of the
continuous transport of the basic cycles $(\beta_1,\beta_2)$ when
the singular line $C$ is crossed. $\gamma_1$ and
$\gamma_2$ are respectively represented in dashed and solid
lines.}
\end{center}
\end{figure}
\begin{figure}
\begin{center}
\includegraphics[scale=0.6]{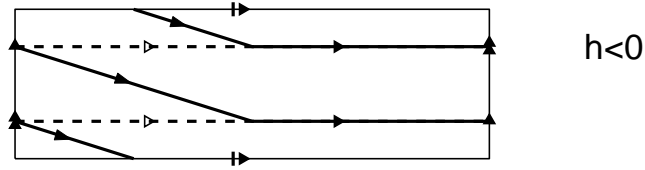}
\caption{\label{figapp2} Schematic representation of the basis
cycles for $h<0$. Representatives of $[\beta_1]$ and $[2\beta_2]$
are respectively represented in dashed and solid lines.}
\end{center}
\end{figure}
Let $\beta_1$ and $\beta_2$ be the two cycles associated to the
flows of $X_1$ and $X_2$ defined in Eqs. (\ref{fhm2a}). They are
the representatives of the classes of homology $[\beta_1]$ and
$[\beta_2]$ which form a basis of $H_1(T^2(h,j),\Z)$ where
$(h,j)\in \mathcal{R}_{reg}$. Only a subgroup of $H_1$ can be
transported continuously across the line $C$. We assume that a
basis for this subgroup is given by $[\beta_1]$ and $[2\beta_2]$.
$2\beta_2$ corresponds to the cycle $\beta_2$ covered twice. For
$h>0$, as representatives of $[\beta_1]$ and $[2\beta_2]$ we
consider the two cycles $\gamma_1=\beta_1$ and $\gamma_2$. The
cycle $\gamma_2$ is the union of two cycles $\beta_2$ with
starting points belonging to the same orbit of the flow of $X_1$
but separate by an angle $\theta=\pi$. To make a link with
analytical calculations, $\Theta$ is taken to be $\pi/2$ in this
case. $\gamma_1$ can be easily transported across $C$ and it
remains unchanged. To transport continuously $\gamma_2$, the two
cycles forming $\gamma_2$ have to be connected in one point in
$h=0$ and then merged to form only one cycle covered once for
$h<0$. Representatives of the basis of $H_1(F^{-1}(h,j),\Z)$ for
$h<0$ are given in Fig. \ref{figapp2}, note that $\Theta=-\pi/2$
since $\Theta$ has a
 discontinuity of size $\pi$ on $C$. Comparison of Fig. \ref{figapp1}
($h<0$) and Fig. \ref{figapp2} leads to the conclusion that after
one loop, the cycle $\gamma_2$ becomes a representative of the
equivalence class $[2\beta_2-\beta_1]$ and that the corresponding
monodromy matrix written formally in the basis
$([\beta_1],[\beta_2])$ is given by
\begin{eqnarray} \label{eqapp1}
M= \left( \begin{array}{cc}
1 & 0 \\
-1/2 & 1
\end{array} \right) \ .
\end{eqnarray}
\subsection{1:-n resonance} \label{appmn}
The geometric construction of Sec. \ref{app12} can be
straightforwardly generalized to other resonances. We consider the
resonance 1:-3 but other resonances can be treated along the same
lines. The 3-curled torus is represented in Fig. \ref{figapp3}
($h=0$) by three rectangles glued along a common edge. Note also
the particular identification of the vertical edges. Following
notations of Sec. \ref{app12}, a basis of the subgroup of
$H_1(F^{-1}(h,j),\Z)$ which can be transported continuously across
$C$ is given by $[\beta_1]$ and $[3\beta_2]$. For $h>0$, as a
representative of $[3\beta_2]$ we choose three cycles $\beta_2$
whose starting points belong to the same orbit of the flow of
$X_j$ with an angle $\theta$ of $2\pi/3$ between each other. We
also assume that $\Theta=\pi/3$ in this case. After crossing $C$,
this cycle belongs for $h<0$ to the equivalence class
$[3\beta_2-\beta_1]$ as shown by Fig. \ref{figapp4}. For $h<0$,
$\Theta=-\pi/3$ since the discontinuity of $\Theta$ on $C$ is
equal to $2\pi/3$. In the basis $([\beta_1],[\beta_2])$, the
monodromy matrix $M$ is given by
\begin{eqnarray} \label{app1}
M= \left( \begin{array}{cc}
1 & 0 \\
-1/3 & 1
\end{array} \right) \ .
\end{eqnarray}
\begin{figure}
\begin{center}
\includegraphics[scale=0.6]{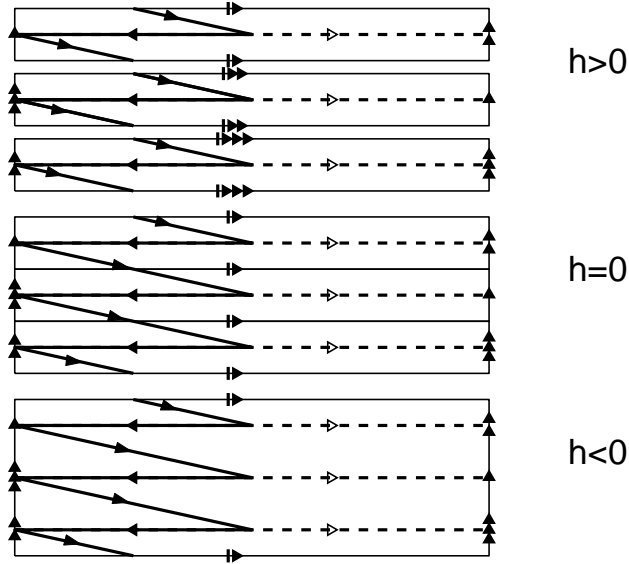}
\caption{\label{figapp3} Same as Fig. \ref{figapp1} but for the
1:-3 resonance.}
\end{center}
\end{figure}
\begin{figure}
\begin{center}
\includegraphics[scale=0.6]{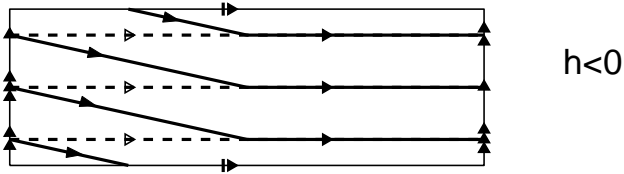}
\caption{\label{figapp4} Same as Fig. \ref{figapp2} but for the
1:-3 resonance.}
\end{center}
\end{figure}
\section{The semi-classical point of view} \label{semi}
We illustrate in this section the relation between classical
monodromy and its semi-classical counterpart for $1:-n$ and $m:-n$
resonant systems. To our knowledge, this point has not been
discussed up to now in the literature. We refer the reader to
Refs. \cite{san1,frac2} for a rigorous definition of this
semi-classical point of view. Here, we consider only the graphical
representations of semi-classical monodromy as a pictorial
illustration of classical monodromy.

We consider two quantum differential operators $\hat{J}$ and
$\hat{H}$ whose classical limits are the Hamiltonians $J$ and
$H$. These two operators commute i.e. $[\hat{J},\hat{H}]=0$. They
thus have a system of common eigenfunctions belonging to $L^2(\R
^2,dq_1\wedge dq_2)$. The corresponding eigenvalues form the
quantum joint spectrum of the energy-momentum map $F$ which is a
2-dimensional lattice of points. We can construct the
quantum-classical bifurcation diagram of $F$ by superimposing both
the quantum joint spectrum and the classical bifurcation diagram.
This has been done for the $1:-n$ and $m:-n$ resonant systems in
Figs. \ref{figapp10} and \ref{figapp11}. Note that the quantum
joint spectrum is defined for a given value $\hbar$ viewed here as
a parameter. Using EBK quantification rules, we also introduce the
semi-classical joint spectrum which is defined as the set of
points $(h,j)\in \mathcal{R}_{reg}$ where the numbers $n_1$ and
$n_2$ given by the relations
\begin{equation} \label{eqsemi1}
\hbar (n_i+\frac{\alpha_i}{4})=\oint _{\gamma_i} \textbf{pdq} \ ,
\end{equation}
are integers. In Eq. (\ref{eqsemi1}), \textbf{pdq} is the
Liouville 1-form and $\alpha_i$ the Maslov index associated to the
cycle $\gamma_i$ where $([\gamma_1],[\gamma_2])$ is a basis of
$H_1(T^2(h,j),\Z)$. We remark that the semi-classical lattice
differs from the quantum lattice by $o(\hbar)$ which is irrelevant
in our study. We also point out that the EBK quantification rules
are not valid near the line of singularities $C$ and have to be
replaced by singular rules\cite{colin3}. In contrast, the quantum
joint spectrum gives a smooth transition of the crossing $C$.
Moreover, locally around a regular value of $\mathcal{R}$, this
lattice is regular in the sense that there exists a map which
sends this lattice to $\hbar\Z^2$ as $\hbar$ tends to zero. A
systematic approach has been developed to check the regularity of
the global spectrum. The method consists in taking a cell i.e. a
quadrilateral whose vertices lie on the points of the lattice,
transporting continuously this cell along a loop $\Gamma$ and
comparing the final cell with the initial one. If the two cells
are different then the system has quantum monodromy. The rotation
matrix which sends the initial cell to the final one is the
quantum monodromy matrix $M_Q$. More precisely, if the cell is
supported by the two vectors $(w_1,w_2)$
 and if after a loop these vectors
are transformed into $(w_1',w_2')$ then $M_Q$ is defined by the
relation
\begin{eqnarray} \label{eqsemi2}
\left( \begin{array}{c}
w'_1 \\
w'_2
\end{array} \right)
=M_Q \left( \begin{array}{c}
w_1 \\
w_2
\end{array} \right) \ .
\end{eqnarray}
Using the semi-classical joint spectrum, it can be shown that
$M_Q=(M_{Cl}^t)^{-1}$ where $M_{Cl}$ is the classical monodromy
matrix. We recall that when the system has fractional monodromy
the size of the cell has to be increased to cross the line of
singularities. In a way analogous to the classical case, a simple
cell cannot be transported continuously across the line $C$. The
multiple cell then becomes the basic cell which is equivalent to
consider only a sublattice of the original lattice. Here, the
vertical lines of the quantum bifurcation diagram which are
parallel to $w_1$ are labeled by $n_1$, the quantum number
associated to $\phi_J$ which is a global quantum number. The cells
are thus multiplied in the other direction. For instance, for the
resonance 1:-3 the size is multiplied by 3. This point is
displayed in Fig. \ref{figapp10}. For a resonance $m:-n$, due to
the form of the monodromy matrix the size is increased from 1 to
$mn$ which explains why we have considered cells of size 6 in Fig.
\ref{figapp11} for the 2:-3 resonance. Examination of Figs.
\ref{figapp10} and \ref{figapp11} shows that in the basis
$(w_1,w_2)$ the quantum monodromy matrices are respectively equal
to
\begin{eqnarray} \label{eqsemi3}
\left( \begin{array}{cc} 1 & 1/3 \\
0 & 1
\end{array} \right) \ ,
\end{eqnarray}
for the 1:-3 resonance and
\begin{eqnarray} \label{eqsemi4}
\left( \begin{array}{cc} 1 & 1/6 \\
0 & 1
\end{array} \right) \ ,
\end{eqnarray}
for the 2:-3 resonance.
\begin{figure}
\begin{center}
\includegraphics[scale=0.4]{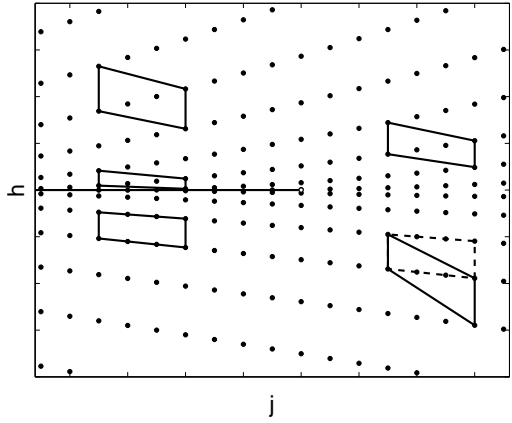}
\caption{\label{figapp10} Semi-classical bifurcation diagram for
the resonance $1:-3$ and the energy-momentum map of Eqs.
(\ref{res1n5}). The line of singularities $C$ is represented by a
solid line. The open dot indicates the position of the origin of
the bifurcation diagram. The final cell after a counterclockwise
closed loop around the origin is depicted in dashed lines.}
\end{center}
\end{figure}
\begin{figure}
\begin{center}
\includegraphics[scale=0.4]{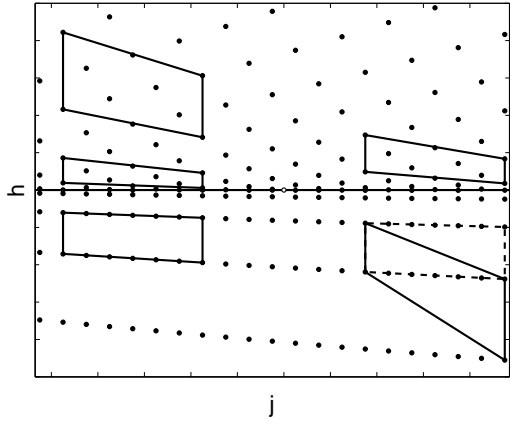}
\caption{\label{figapp11} Same as Fig. \ref{figapp10} but for the
2:-3 resonance. The corresponding energy-momentum map is given by
Eqs. (\ref{eqsemi5}).}
\end{center}
\end{figure}
The energy-momentum map used for the 2:-3 resonance is defined by
\begin{eqnarray} \label{eqsemi5}
F=\left\{ \begin{array}{ll} J \\
\pi_3+(\pi_1+J)(\pi_1-J) \\
\end{array} \right. \ .
\end{eqnarray}
\section{Reduction in the complex approach}\label{appred}
We detail in this section the reduction in the complex approach
for an $m:-n$ resonant system. This reduction is different from
the reduction used in the real approach which is associated to the
$S^1$-action of the flow of the Hamiltonian $J$. Both reductions
leave invariant the polynomials $(J,\pi_1,\pi_2,\pi_3)$ defined by
Eqs. (\ref{res1n3}).

We start from the complexified
phase space $T^*\C^2$ with $(p_1,q_1,p_2,q_2)\in \C^4$.
The reduction is based on two $SO(2,\C)$ actions $\Phi_1$ and $\Phi_2$. We recall that a matrix $R\in SO(2,\C)$ is a $2\times 2$ matrix which reads
\begin{eqnarray} \label{eqred2}
\left( \begin{array}{cc} a & -b \\
b & a
\end{array} \right) \ ,
\end{eqnarray}
where $(a,b)\in \C^2$ and $a^2+b^2=1$.
$\Phi_k$ ($k\in\{1,2\}$) is a map from $SO(2,\C)\times \C^2$ to $\C^2$ which
associates to each couple $\big(R,(q_k,p_k)\big)$ the point of
coordinates
\begin{eqnarray} \label{eqred3}
\left( \begin{array}{cc} a_k & -b_k \\
b_k & a_k
\end{array} \right)\left( \begin{array}{c} q_k \\
p_k
\end{array} \right)
=\left( \begin{array}{cc} a_kq_k-b_kp_k \\
b_kq_k+a_kp_k
\end{array} \right)
\ .
\end{eqnarray}
We next introduce new coordinates which can be written as
follows
\begin{eqnarray} \label{eqred4}
\left\{ \begin{array}{llll} \eta_1=q_1-ip_1 \\
\xi_1=q_1+ip_1 \\
\eta_2=q_2-ip_2 \\
\xi_2=q_2+ip_2 \\
\end{array} \right. \ .
\end{eqnarray}
Under the action of $\Phi_k$, these new coordinates transform into
$\lambda_k\xi_k$ and $\lambda_k^{-1}\eta_k$
 where $\lambda_k=a_k+ib_k$. Finally, simple algebra shows that
the complexified invariant polynomials $(J,\pi_1,\pi_2,\pi_3)$ are invariant under the conjoint action of
$\Phi_1$ and $\Phi_2$ if $\lambda_1^n\lambda_2^m=1$. One can conclude that the reduction is associated to a $\C^*$-action $\Phi$ from $\C ^*\times \C^4$
to $\C^4$ defined as follows
\begin{equation} \label{eqred5}
\Phi : \big(\lambda,(\xi_1,\eta_1,\xi_2,\eta_2)\big)\to (\lambda^m\xi_1,\lambda^{-m}\eta_1,\lambda^{-n}\xi_2,\lambda^n\eta_2) \ .
\end{equation}

\end{document}